\documentclass[11pt,reqno,letterpaper]{amsart}

\usepackage[T1]{fontenc}
\usepackage{lmodern}
\usepackage{amsmath,amssymb,mathtools}
\usepackage[margin=1in]{geometry}
\usepackage{microtype}
\usepackage{booktabs}
\usepackage{enumitem}
\usepackage{xcolor}
\usepackage{graphicx}
\usepackage{aliascnt}
\usepackage[boxruled,linesnumbered,noend]{algorithm2e}
\DontPrintSemicolon
\SetKwInput{KwIn}{Input}
\SetKwInput{KwOut}{Output}
\SetAlCapNameFnt{\normalfont}
\SetAlCapFnt{\bfseries}
\SetAlgoNlRelativeSize{-1}
\IncMargin{0.75em}
\SetCommentSty{textnormal}
\SetKwComment{tcp}{\(\triangleright\) }{}
\SetAlgoInsideSkip{smallskip}
\SetAlgoSkip{medskip}
\usepackage[colorlinks=true,linkcolor=blue!55!black,
  citecolor=blue!55!black,urlcolor=blue!55!black]{hyperref}
\usepackage[nameinlink,noabbrev,capitalise]{cleveref}
\crefname{algocf}{Algorithm}{Algorithms}
\Crefname{algocf}{Algorithm}{Algorithms}

\usepackage{tikz}
\usetikzlibrary{arrows.meta,fit}

\newtheorem{theorem}{Theorem}[section]
\newaliascnt{lemma}{theorem}
\newtheorem{lemma}[lemma]{Lemma}
\aliascntresetthe{lemma}

\newaliascnt{proposition}{theorem}
\newtheorem{proposition}[proposition]{Proposition}
\aliascntresetthe{proposition}

\newaliascnt{corollary}{theorem}

\aliascntresetthe{corollary}

\theoremstyle{remark}
\newaliascnt{remark}{theorem}
\newtheorem{remark}[remark]{Remark}
\aliascntresetthe{remark}

\newcommand{\eps}{\varepsilon}
\newcommand{\TV}{\mathrm{TV}}
\newcommand{\vbl}{\operatorname{vbl}}
\newcommand{\Good}{\operatorname{Good}}
\newcommand{\Law}{\operatorname{Law}}
\newcommand{\dist}{\operatorname{dist}}
\newcommand{\norm}[1]{\left\lVert#1\right\rVert}

\newcommand{\Count}{\ensuremath{\mathsf{Count}}}
\newcommand{\Sample}{\ensuremath{\mathsf{Sample}}}
\newcommand{\Insert}{\ensuremath{\mathsf{Insert}}}
\newcommand{\EstimateMarginal}{\ensuremath{\mathsf{EstimateMarginal}}}
\numberwithin{equation}{section}

\setlist[itemize]{leftmargin=1.5em}
\setlist[enumerate]{leftmargin=1.7em}
\allowdisplaybreaks[1]
\hypersetup{pdftitle={Approximate Counting and Sampling under an Asymmetric Lovasz Local Lemma}}

\title{A general counting and sampling  Lov\'asz local lemma}
\author[V. Jain]{Vishesh Jain}
\address{Department of Mathematics, Statistics, and Computer Science, University of Illinois Chicago, Chicago, IL, 60607 USA}
\email{visheshj@uic.edu}

\author[C. Mizgerd]{Clayton Mizgerd}
\address{Department of Mathematics, Statistics, and Computer Science, University of Illinois Chicago, Chicago, IL, 60607 USA}
\email{cmizge2@uic.edu}

\author[H.T. Pham]{Huy Tuan Pham}
\address{Department of Mathematics, University of Chicago, Chicago, IL 60637 USA}
\email{htpham@uchicago.edu}
\date{}

\begin{document}
\begin{abstract}
Consider a constraint satisfaction problem $\mathcal C$ on finitely many independent random variables with dependency graph $G$. Let $p_a$ be the violation probability of a constraint $a\in \mathcal C$ and $N_G^2 (a)$ the set of constraints at distance one or two from $a$ in $G$. Suppose that, there exists $x\in (0,1)^{\mathcal C}$ such that, for a sufficiently small universal constant $c > 0$, and for all $a \in \mathcal C$,
\[
p_a \leq c \cdot x_a \prod_{b\in N_G^2(a)}(1-x_b).
\]
Under the above analog of the asymmetric Lov\'asz Local Lemma, we give an FPRAS for the probability that all constraints are satisfied, and an approximate sampler, running in polynomial expected time, for the product distribution conditioned on this event. The degree of the polynomial in the running time is independent of the domain sizes, constraint sizes, or degree of the dependency graph. Up to the choice of the constant $c$, our condition on $p_a$ matches known hardness results. 

Our work builds on the method of Liu, Wang, Yin, Zhang, and Zhou, who obtained an FPRAS for the probability of satisfaction in the setting of the symmetric Lov\'asz Local Lemma. Our sampling result is new even in this special case.  
\end{abstract}
\maketitle

\section{Introduction}\label{sec:introduction}

The Lov\'asz Local Lemma (LLL) is a fundamental tool in probabilistic combinatorics which gives a sufficient condition for avoiding
a family of bad events in terms of their probabilities and dependencies. We work in the variable setting. Let $(X_v)_{v\in V}$ be a finite family
of independent random variables, where $X_v$ has law $\nu_v$ with finite
support $\Omega_v$, and write
\[
\Omega=\prod_{v\in V}\Omega_v,
\qquad
\nu=\bigotimes_{v\in V}\nu_v
\]
for the product space and product distribution, respectively. 
For $U\subseteq V$, write $\Omega_U=\prod_{v\in U}\Omega_v$ and
$\nu_U=\bigotimes_{v\in U}\nu_v$. Let $\mathcal C$ be a finite family of constraints on these variables.
Each constraint $a\in\mathcal C$ is given by a set of variables
$\vbl(a)\subseteq V$ together with a collection of forbidden
assignments to those variables. Let $B_a$ denote the ``bad'' event that constraint $a$ is violated, and write
$p_a=\nu(B_a)$ for its probability. The dependency graph $G$ has vertex set $\mathcal C$, with two distinct
constraints adjacent if they share a variable. Write $N_G(a)$ for the
neighbors of $a$ in $G$.

The (asymmetric) LLL \cite{erdosLovasz1975} guarantees that the probability
of satisfying all constraints is positive whenever there is
$x\in(0,1)^{\mathcal C}$ such that
\[
p_a\leq x_a\prod_{b\in N_G(a)}(1-x_b)
\qquad(a\in\mathcal C).
\] The original proof of the LLL is nonconstructive. In a landmark work, Moser and Tardos
\cite{moserTardos2010} gave an efficient algorithm for finding a
satisfying assignment in the variable setting by repeatedly resampling
the variables of a violated constraint.

In this paper, we are interested in efficient algorithms for
approximating the probability that all constraints are satisfied
and sampling from the product distribution conditioned on this event. When the coordinate distributions are uniform, these are the
problems of approximately counting satisfying assignments and sampling
an almost-uniform satisfying assignment. Both of these problems have been studied extensively, primarily for uniform product distributions in the symmetric setting. Starting from the work of Moitra~\cite{moitra2019}, and following a long line of works by multiple researchers~(see~\cite{liuWangYinZhangZhou2026} and the references therein), Liu, Wang, Yin, Zhang, and Zhou (LWYZZ)
\cite{liuWangYinZhangZhou2026} recently gave an FPRAS for the probability
that all constraints are satisfied when the coordinate distributions
are uniform and
\[
4\mathrm e\,p(\Delta+1)^2\leq1.
\]
Here $p$ is an upper bound on the violation probabilities, and
$\Delta\geq1$ is an upper bound on the maximum degree of $G$. Up to a universal constant, this condition matches the hardness result
of Bez\'akov\'a et al.~\cite{bezakovaEtAl2019}, assuming
$\mathrm{NP}\ne\mathrm{RP}$.

On the other hand, LWYZZ leave open whether efficient approximate sampling is also possible
under a condition of the form $p\Delta^2\leq c$. For general constraints in this setting, the previous best result is due to He, Wang, and Yin
\cite{heWangYin2022}, who gave an approximate sampler
under the condition
\[
kq^2p\Delta^5\leq c.
\]
Here, each constraint involves at most $k$ variables, each variable has
at most $q$ possible values, and $c>0$ is a sufficiently small universal
constant. For atomic constraints (these are constraints forbidding a single assignment),
Wang and Yin \cite{wangYin2024} obtained an approximate sampler under
the condition
\[
p\Delta^{2+o_{q_{\min}}(1)}\lesssim1,
\]
where $q_{\min}$ is the minimum domain size. 

\subsection{Our results}

We prove a counting and sampling analogue of the general asymmetric Lov\'asz
Local Lemma in the variable setting. In order to state our result, we need some notation. Let $G^2$ be the square of the graph $G$, in which two vertices are adjacent
if their distance in $G$ is one or two, and write $N_G^2(a)$ for the
neighbors of $a$ in $G^2$. For $\mathcal D\subseteq\mathcal C$, let $\Good(\mathcal D)$ denote
the set of assignments satisfying every constraint in $\mathcal D$,
and write
\[
Z_{\mathcal D}=\nu(\Good(\mathcal D)).
\]
When $Z_{\mathcal D}>0$, let
$\mu_{\mathcal D}$ denote $\nu$ conditioned on $\Good(\mathcal D)$.

\begin{theorem}\label{thm:main}
There exists a universal constant $c>0$ such that the following holds.
Suppose that
\begin{equation}\label{eq:full-square}
p_a\leq c x_a\prod_{b\in N_G^2(a)}(1-x_b)
\qquad(a\in\mathcal C)
\end{equation}
for some $x\in(0,1)^{\mathcal C}$. Then $Z_{\mathcal C}>0$, and there
is an FPRAS for $Z_{\mathcal C}$. Moreover, for every $0<\eps\leq1$,
there is a randomized algorithm, running in polynomial expected time,
which returns $\sigma\in\Good(\mathcal C)$ with
\[
\norm{\Law(\sigma)-\mu_{\mathcal C}}_{\TV}\leq\eps.
\]
\end{theorem}

\begin{remark} For our proof, $c = 1/100$ is sufficient; we have made no attempt to optimize this. We use the same basic computational assumptions as in the Moser--Tardos
algorithm \cite{moserTardos2010}. The domains, the sets $\vbl(a)$,
and their incidence lists are given explicitly. Write $I$ for the
encoded instance size. We assume access to
oracles for testing constraints and drawing independent samples
from the coordinate distributions, each with polynomial
cost in $I$. For every $0<\eps,\delta<1$, the counting algorithm estimates
$Z_{\mathcal C}$ to relative error $\eps$ with probability at least
$1-\delta$. Its worst-case running time is polynomial in $I$,
$1/\eps$, $\log(1/\delta)$. The sampler has expected
running time polynomial in $I$ and $1/\eps$. The degree of the polynomial is independent of the domain sizes, constraint sizes,
or degree of the dependency graph.
\end{remark}

\medskip

\paragraph{\bf Independent work.} After all the results in our work had been obtained and we were in the final stages of preparing the manuscript, we learned of independent and concurrent work of Achlioptas~\cite{achlioptas2026sampling}. The main result of~\cite{achlioptas2026sampling} is an approximate sampling and counting algorithm in the setting of \cref{thm:main}, which runs in polynomial time provided that $D$ (the maximum domain size of the variables), $k$ (the maximum arity of a constraint), and $\Delta$ (the maximum degree of the dependency graph) are all bounded; specifically, the running time is of the form $(|I|/\varepsilon)^{O(k\Delta \log D)}$. In this special setting (in fact, without assuming any bound on $k$ or $D$), one may use a much simpler version of the sampler in our main~\cref{thm:main} with running time  $(|I|/\varepsilon)^{O(\log \Delta)}$; in contrast, our main sampler runs in expected time $(|I|/\varepsilon)^{O(1)}$. Since the bounded-$\Delta$ setting serves as an instructive warm-up for our main construction, we discuss this in more detail in the proof overview below (see~\cref{sub:bounded-degree-warmup}).

\subsection{Proof overview}

We restrict ourselves to the symmetric setting in this overview since it already contains the key ideas.  Therefore, suppose
that every constraint has violation probability at most $p$, the
dependency graph has maximum degree $\Delta\geq1$, and
$p(\Delta+1)^2 \leq c$ for a sufficiently small universal constant $c > 0$. The sampler in the general asymmetric setting is similar to the one in the growing $\Delta$ case (see~\cref{sec:unbounded-degree}), with a similar analysis provided that one uses the branching-process arguments from Moser--Tardos \cite{moserTardos2010} in the appropriate places. 

Fix $a\in\mathcal D\subseteq\mathcal C$, and write
$\mathcal D^-=\mathcal D\setminus\{a\}$. Both counting and
sampling can be reduced to adding one constraint at a time; this approach was introduced in~\cite{wangYin2024}. For counting, as in~\cite{liuWangYinZhangZhou2026}, we estimate the ratios
\[
\frac{Z_{\mathcal D}}{Z_{\mathcal D^-}}
=1-r_{\mathcal D,a},
\qquad
r_{\mathcal D,a}=\mu_{\mathcal D^-}(B_a),
\]
and multiply them. For
sampling, we provide an algorithm which converts a sample from $\mu_{\mathcal D^-}$ into an
approximate sample from $\mu_{\mathcal D}$.

\subsubsection{\bf The expansion of LWYZZ}
A $2$-tree in a graph is an independent set that is connected
in the square of the graph. This notion goes back to
Alon's work on the algorithmic Lov\'asz Local Lemma
\cite{alon1991} and has since been used extensively in the literature on counting and sampling analogues of the Lov\'asz Local Lemma. 
The usefulness of this notion comes from two properties. The constraints
in a $2$-tree involve disjoint variables, so the probability
that all of them are violated is at most $p^{|\mathcal T|}$.
Moreover, the number of $2$-trees of size $t$ containing a
given constraint is at most $[O((\Delta+1)^2)]^{t-1}$.
Thus, when $p(\Delta+1)^2$ is sufficiently small, the decrease
in probability compensates for the number of $2$-trees.

The key ingredient in LWYZZ is an exact expansion of $r_{\mathcal D,a}$ in terms of
$2$-trees. Write $\mathfrak T_{\mathcal D,a}$ for the
$2$-trees containing $a$ in the induced graph $G_{\mathcal D}=G[\mathcal D]$.
For each $\mathcal T\in\mathfrak T_{\mathcal D,a}$, let
$\mathcal D_{\mathcal T}$ consist of the constraints at
distance greater than two from $\mathcal T$ in $G_{\mathcal D}$. LWYZZ showed that
\[
r_{\mathcal D,a}
=\sum_{\mathcal T\in\mathfrak T_{\mathcal D,a}}
(-1)^{|\mathcal T|-1}\theta_{\mathcal T}
\frac{Z_{\mathcal D_{\mathcal T}}}{Z_{\mathcal D^-}},
\]
where $\theta_{\mathcal T}$ is the probability of an event
$E_{\mathcal T}$ requiring all constraints in $\mathcal T$
to be violated and certain neighbouring constraints to
be satisfied. In particular,
$\theta_{\mathcal T}\leq p^{|\mathcal T|}$.
Write $U_{\mathcal T}$ for the variables appearing
in $\mathcal T$ or its neighbours in $G_{\mathcal D}$.
The event $E_{\mathcal T}$ depends only on
$U_{\mathcal T}$, and no constraint in
$\mathcal D_{\mathcal T}$ involves these variables. This separation
will be important for our sampling algorithm.

The ratios in the expansion increase the contribution of
each tree, but the local lemma bounds this increase by (see~\cite{haeuplerSahaSrinivasan2011})
\[
\frac{Z_{\mathcal D_{\mathcal T}}}{Z_{\mathcal D^-}}
\leq \exp\bigl(O(p(\Delta+1)^2|\mathcal T|)\bigr),
\]
since at most $O((\Delta+1)^2|\mathcal T|)$ constraints
are deleted. Combining this bound with
$\theta_{\mathcal T}\leq p^{|\mathcal T|}$ and the bound
on the number of $2$-trees, we obtain
\[
\sum_{\substack{\mathcal T\in\mathfrak T_{\mathcal D,a}\\
|\mathcal T|=t}}
\theta_{\mathcal T}
\frac{Z_{\mathcal D_{\mathcal T}}}{Z_{\mathcal D^-}}
\leq
[C(\Delta+1)^2]^{t-1}p^t
\exp\bigl(C'p(\Delta+1)^2t\bigr).
\]
For sufficiently small $p(\Delta+1)^2$, this bound
decreases geometrically with $t$.

Order the constraints of
$\mathcal D^-\setminus\mathcal D_{\mathcal T}$ as
$b_{\mathcal T,1},\ldots,b_{\mathcal T,\ell(\mathcal T)}$,
and set
$\mathcal D_{\mathcal T,i}
=\mathcal D_{\mathcal T}\cup
\{b_{\mathcal T,1},\ldots,b_{\mathcal T,i}\}$
for $0\leq i\leq\ell(\mathcal T)$. Then
\[
\frac{Z_{\mathcal D_{\mathcal T}}}{Z_{\mathcal D^-}}
=\prod_{i=1}^{\ell(\mathcal T)}
\frac{Z_{\mathcal D_{\mathcal T,i-1}}}
     {Z_{\mathcal D_{\mathcal T,i}}}
=\prod_{i=1}^{\ell(\mathcal T)}
(1-r_{\mathcal D_{\mathcal T,i},b_{\mathcal T,i}})^{-1}.
\]
Since each $\mathcal D_{\mathcal T,i}\subseteq
\mathcal D^-\subsetneq\mathcal D$, the expansion gives
a recursive formula for $r_{\mathcal D,a}$.
LWYZZ use a randomized evaluation of this formula to obtain
an unbiased estimator with bounded variance and polynomial
expected running time. Our counting
algorithm follows this approach, with appropriate modifications to handle the asymmetric setting; we describe it in
\cref{sec:estimates,sec:counting}.

\subsubsection{\bf A simple sampler for the bounded-degree setting.}
\label{sub:bounded-degree-warmup}
As a warm-up, we discuss a substantially simpler sampler when  $\Delta$ is fixed. The starting point is the following analogue of LWYZZ for conditional distributions
\[
\mu_{\mathcal D^-}-\mu_{\mathcal D}
=\sum_{\mathcal T}\lambda_{\mathcal T}
\left[
(\mathbf1_{E_{\mathcal T}}-\theta_{\mathcal T})
\mu_{\mathcal D_{\mathcal T}}
+\theta_{\mathcal T}
(\mu_{\mathcal D_{\mathcal T}}-\mu_{\mathcal D^-})
\right],
\]
where
$\lambda_{\mathcal T}
=(-1)^{|\mathcal T|-1}Z_{\mathcal D_{\mathcal T}}/Z_{\mathcal D}$. Using the same telescoping argument as above, we recursively
expand the second summand until the sum of the sizes of
the $2$-trees chosen along each branch reaches a threshold $L$. This gives
\[
\mu_{\mathcal D^-}-\mu_{\mathcal D}
=\sum_{j=1}^{J}\alpha_j(\mathbf1_{E_j}-\theta_j)\mu_{\mathcal D_j}
+\tau_L.
\]
Here $E_j$ is an event depending on variables
$U_j$ disjoint from $\vbl(\mathcal D_j)$, and
$\theta_j=\nu(E_j)$. Using the same tree-counting and local lemma estimates
as above, we obtain $J=(\Delta+1)^{O(L)}$,
$|\alpha_j|\leq\exp(O(L))$, and
$\norm{\tau_L}_{\TV}\leq\exp(-\Omega(L))$. Taking $L=O(\log(2/\delta))$, we can enumerate all the
pairs $(\mathcal D_j,U_j)$ in time
$I^{O(1)}(2/\delta)^{O(\log(\Delta+1))}$,
which is polynomial in $I$ and $1/\delta$ for fixed $\Delta$.

We use these pairs to construct a simple Markov chain on
the state space
$\Good(\mathcal D)\cup\bigcup_{j=1}^{J}\Good(\mathcal D_j)$. Let $Q_j$
resample the variables in $U_j$ from their product distribution,
leaving the remaining variables fixed. For a parameter
$\beta>0$, the transition kernel is
\[
K_0(\sigma,\cdot)
=
\frac{
\mathbf1_{\Good(\mathcal D)}(\sigma)\,\mathrm{Id}(\sigma,\cdot)
+\beta\sum_{j=1}^{J}
\mathbf1_{\Good(\mathcal D_j)}(\sigma)\,Q_j(\sigma,\cdot)
}{
\mathbf1_{\Good(\mathcal D)}(\sigma)
+\beta\sum_{j=1}^{J}\mathbf1_{\Good(\mathcal D_j)}(\sigma)
}.
\]
Thus, each $Q_j$ has weight
$\beta\mathbf1_{\Good(\mathcal D_j)}(\sigma)$, and staying
put has weight $\mathbf1_{\Good(\mathcal D)}(\sigma)$. Our sampler starts from $\mu_{\mathcal D^-}$ and runs this
chain for $N$ steps. It returns the final assignment if
it satisfies $\mathcal D$, and a fixed satisfying assignment
otherwise.

The kernel $K_0$ is reversible with stationary distribution
\[
\pi_0
=
\frac{
Z_{\mathcal D}\mu_{\mathcal D}
+\beta\sum_{j=1}^{J}Z_{\mathcal D_j}\mu_{\mathcal D_j}
}{
Z_{\mathcal D}+\beta\sum_{j=1}^{J}Z_{\mathcal D_j}
}.
\]
The local lemma bounds
$Z_{\mathcal D_j}/Z_{\mathcal D}\leq\exp(O(L))$, so
\[
\norm{\pi_0-\mu_{\mathcal D}}_{\TV}
\leq
\beta\sum_{j=1}^{J}\frac{Z_{\mathcal D_j}}{Z_{\mathcal D}}
\leq\beta J\exp(O(L)).
\]
Thus we can make this distance at most $\delta$ by choosing
$\beta^{-1}=(2/\delta)^{O(\log(\Delta+1))}$, which is
polynomial in $1/\delta$ for fixed $\Delta$.

For convergence, the stationarity of $\pi_0$ and the data-processing inequality give
\[
\norm{(\mu_{\mathcal D^-})K_0^N-\mu_{\mathcal D}}_{\TV}
\leq
\norm{(\mu_{\mathcal D^-}-\mu_{\mathcal D})K_0^N}_{\TV}
+2\norm{\pi_0-\mu_{\mathcal D}}_{\TV}.
\]
The second term is already small. We control the first term using the above decomposition of $\mu_{\mathcal D^-} - \mu_{\mathcal D}$ and the resampling operations defining the chain. Indeed, for any $f$, since $\mathbf1_{E_j}-\theta_j$ is orthogonal to $Q_jf$ in
$L^2(\mu_{\mathcal D_j})$, Cauchy--Schwarz gives
\[
\bigl|(\mu_{\mathcal D^-}-\mu_{\mathcal D})(f)\bigr|
\leq B_L\sqrt{\mathcal E_{K_0}(f,f)}
+2\norm{\tau_L}_{\TV}\norm f_\infty,
\]
where, writing $w_j$ for the mixture weights in $\pi_0$,
\[
B_L^2=\sum_j\frac{\alpha_j^2}{w_j} = (2/\delta)^{O(\log(\Delta+1))},
\qquad
\mathcal E_{K_0}(f,f)
=\sum_jw_j\norm{f-Q_jf}_{L^2(\mu_{\mathcal D_j})}^2.
\]
Finally, since $K_0$ is reversible and positive semidefinite, its spectral decomposition gives an upper bound on the Dirichlet form
\[
\mathcal E_{K_0}(K_0^Nf,K_0^Nf)
\leq \max_{\lambda}(1-\lambda)\lambda^{2N} \le \frac1{2N+1}
\qquad(\norm f_\infty\leq1).
\]
Taking $N=(2/\delta)^{O(\log(\Delta+1))}$ sufficiently
large therefore gives accuracy $\delta$, with total
running time polynomial in $I$ and $1/\delta$ for
fixed $\Delta$.  To get total variation distance $\varepsilon$, we will need to take $\delta = \varepsilon/|\mathcal C|$.

\subsubsection{\bf Handling unbounded $\Delta$}
\label{sec:unbounded-degree}

We continue to work in the symmetric setting, but now, $\Delta$ is allowed to grow with the size of the instance. Note that even if one step of the kernel $K_0$ in the above paragraph could be implemented efficiently, the bound on the number of steps still involves
$B_L^2=\sum_j\alpha_j^2/w_j$, which depends on the number of terms and is of order $(2/\delta)^{O(\log(\Delta+1))}$. To get around this, we construct a different chain on a larger
state space, with separate copies of
$\Good(\mathcal D^-)$ and $\Good(\mathcal D)$ and
auxiliary states connecting them. Each step of this chain can be simulated in polynomial expected time. To bound the number of steps needed for convergence, we use the recursive expansion of $\mu_{\mathcal D^-} - \mu_{\mathcal D}$ to construct a signed flow
from $\mu_{\mathcal D^-}$ on the input copy to
$\mu_{\mathcal D}$ on the output copy, and bound its
energy with respect to the chain's edge conductances.

We first recall the standard way in which a flow of bounded energy bounds convergence of a Markov chain. For a lazy reversible kernel $K$ with stationary law $\pi$,
a signed flow assigns a real number $\varphi(X,Y)$
to each oriented edge, with
$\varphi(Y,X)=-\varphi(X,Y)$.
Its divergence is the net outflow at each state:
\[
\partial\varphi(X)=\sum_Y\varphi(X,Y).
\]
In our setting, if $\partial\varphi=\mu_{\mathcal D^-}-\mu_{\mathcal D}$,
where the two laws are supported on the separate copies $\Good(\mathcal D^-)$ and $\Good(\mathcal D)$, Cauchy--Schwarz gives
\[
\bigl|(\mu_{\mathcal D^-}-\mu_{\mathcal D})(f)\bigr|^2
\leq
\left(
\sum_{\{X,Y\}}
\frac{\varphi(X,Y)^2}{\pi(X)K(X,Y)}
\right)\mathcal E_K(f,f);
\]
see \cite[Section~9.4]{levinPeresWilmer2017}. The energy of the flow is the sum in the parentheses. If the energy
is at most $A$, applying the inequality to $K^Nf$
and using
$\mathcal E_K(K^Nf,K^Nf)\leq1/(2N+1)$ for
$\norm f_\infty\leq1$ gives
\[
\norm{\mu_{\mathcal D^-}K^N-\mu_{\mathcal D}}_{\TV}
\leq
\frac{\sqrt A}{2\sqrt{2N+1}}
+2\norm{\pi-\mu_{\mathcal D}}_{\TV}.
\]
Thus it suffices to bound the energy and ensure that the
stationary law $\pi$ is close to the target $\mu_{\mathcal{D}}$. 

To motivate the chain, we first describe the flow. Recall the identity
\begin{equation} \label{eq:expansion-of-mu}
\mu_{\mathcal D^-}-\mu_{\mathcal D}
=\sum_{\mathcal T}\lambda_{\mathcal T}
\left[
(\mathbf1_{E_{\mathcal T}}-\theta_{\mathcal T})
\mu_{\mathcal D_{\mathcal T}}
+\theta_{\mathcal T}
(\mu_{\mathcal D_{\mathcal T}}-\mu_{\mathcal D^-})
\right].
\end{equation}
We work inductively on $|\mathcal{D}|$, so assume we have constructed a flow for all smaller insertion problems. Fix a tree $\mathcal T$.
The two differences inside the brackets in
\eqref{eq:expansion-of-mu} correspond to the two arrows
\[
\mathbf1_{E_{\mathcal T}}\mu_{\mathcal D_{\mathcal T}}
\ \longrightarrow\
\theta_{\mathcal T}\mu_{\mathcal D_{\mathcal T}}
\ \longrightarrow\
\theta_{\mathcal T}\mu_{\mathcal D^-}.
\]
For each
$\sigma\in E_{\mathcal T}\cap\Good(\mathcal D_{\mathcal T})$,
create a tree state $\mathsf T_{\mathcal T}(\sigma)$.
The first arrow resamples $U_{\mathcal T}$ from its
product distribution. Recall that $E_{\mathcal T}$
depends only on these variables, whereas no constraint
in $\mathcal D_{\mathcal T}$ involves them.
Therefore, the resampling sends
$\mathbf1_{E_{\mathcal T}}\mu_{\mathcal D_{\mathcal T}}$
to $\theta_{\mathcal T}\mu_{\mathcal D_{\mathcal T}}$.
Sending the first measure along the resampling edges
defines a flow between these two measures.

For the second arrow, we add the constraints of
$\mathcal D^-\setminus\mathcal D_{\mathcal T}$ one
at a time. Introduce a \emph{layer} consisting of
a copy of $\Good(\mathcal D_{\mathcal T,i})$ for
each $i=0,\ldots,\ell(\mathcal T)$.
The resampling edges lead to layer $0$.
By induction, smaller insertion chains provide
flows between consecutive layers; we multiply
these flows by $\theta_{\mathcal T}$.
At each intermediate layer, the preceding flow
has divergence $-\theta_{\mathcal T}
\mu_{\mathcal D_{\mathcal T,i}}$, and the next has
divergence $+\theta_{\mathcal T}
\mu_{\mathcal D_{\mathcal T,i}}$.
Thus the total divergence in the intermediate layers is zero.

The layers and smaller chains form a \emph{corridor}
carrying a flow from
$\theta_{\mathcal T}\mu_{\mathcal D_{\mathcal T}}$
to $\theta_{\mathcal T}\mu_{\mathcal D^-}$.
We connect its last layer to the input copy.
Finally, we multiply both the resampling flow
and the corridor flow by $\lambda_{\mathcal T}$,
reversing their directions when this coefficient
is negative.

The problem with this construction is that the tree
states are still sources or sinks, depending on the sign of $\lambda_{\mathcal T}$, whereas we want
nonzero divergence only on the input and output
copies. The selection procedure used to derive the $2$-tree
expansion gives canonical pairs
$\mathsf T_{\mathcal T}(\sigma)
\leftrightarrow\mathsf T_{\mathcal T'}(\sigma)$
with equal and opposite divergences. Joining these pairs, together with connections
from the unpaired tree states to the input and
from the input to the output, gives the required
divergence $\mu_{\mathcal D^-}-\mu_{\mathcal D}$.

We now turn this construction into a reversible
Markov chain whose transitions can be simulated efficiently (in particular, without access to partition functions). A first issue is making the resampling step from
a tree state $\mathsf T_{\mathcal T}(\sigma)$ to
the first layer reversible. Reversing the resampling
would appear to require sampling conditioned on
$E_{\mathcal T}$, which is not readily available.
We avoid this by simply recording the initial assignment
$z$ on $U_{\mathcal T}$ and carrying it along the
corridor from $\mathcal D_{\mathcal T}$ to
$\mathcal D^-$. This allows us to trivially undo the resampling step
by returning to the recorded assignment.

The second, more substantial, issue is how to leave
the input copy and choose a corridor. We cannot
afford to enumerate all corridors. Instead, we use
the local sampling procedure from LWYZZ's counting
algorithm \cite[Algorithm~2]{liuWangYinZhangZhou2026}
to generate a random list $\mathbf L$ of corridor
labels, as follows. For each constraint $b\in\mathcal D$, independently
draw an assignment from $\nu_{\vbl(b)}$. These assignments
need not agree on shared variables.
We consider only trees $\mathcal T$ whose constraints
are all violated by their respective samples.
For each such tree, keep the samples on its
disjoint variable sets and sample the remaining
variables in $U_{\mathcal T}$ independently from
their product distribution. Write $z_{\mathcal T}$
for the resulting assignment, using fresh randomness
for each completion. Discard the tree if
$z_{\mathcal T}\notin E_{\mathcal T}$, and let
\[
\mathbf L
=\{(\mathcal T,z_{\mathcal T}):
  z_{\mathcal T}\in E_{\mathcal T}\}
\]
be the list of retained entries. We choose a corridor uniformly from the random list.
To make this move reversible, we use the same list-generation procedure conditioned
to contain the corridor's label, which can be implemented by simply
fixing the corresponding local samples.

Finally, we slow departures from the output copy $\Good(\mathcal{D})$
by a sufficiently small factor $\beta$, making the stationary law
close to $\mu_{\mathcal D}$.

For the running time, using the same bounds on $2$-trees as in the counting
argument, we control both the energy of the flow and the expected
cost of simulating the chain. The convergence bound above
then gives polynomial expected running time in $I$ and
$1/\delta$, with an exponent independent of $\Delta$.
We give the details in
\cref{sec:insertion,sec:insertion-proof,sec:implementation}.

\subsection{Acknowledgments} V.J.~is supported by NSF grant DMS-2237646. C.M.~is supported by a Simons Graduate Dissertation Fellowship. H.T.P.~is supported by a Clay Research Fellowship, a Sloan Research Fellowship and NSF grant DMS-2543870. 

\medskip

\paragraph{\bf Statement on AI use} The first two authors had developed a version of the asymmetric counting and sampling lemma for monotone constraints in Spring, 2025. After the appearance of LWYZZ, it was clear to the authors that their counting estimator could be extended to the general asymmetric setting as well; ChatGPT 5.6 Sol was used to supply the details. More substantively, the sampler in \cref{thm:main} was developed in conversations between the authors and ChatGPT 5.6 Sol. The first correct version of the sampler was extremely complicated and rather hard to understand. It was substantially simplified by the authors, and then building on the authors' work, by ChatGPT 6 Astra. The authors used Codex for assistance with preparing the manuscript. The mathematical content, the final text, and any errors are the responsibility of the authors.

\section{Preliminaries}\label{sec:reduction}

Throughout the proof, we assume \eqref{eq:full-square} with $c=1/100$,
and write $m=|\mathcal C|$. Since this condition implies the usual LLL
condition, $Z_{\mathcal D}>0$ for every $\mathcal D\subseteq\mathcal C$.
For a fixed pair $(\mathcal D,a)$ with $a\in\mathcal D$, write $\mathcal D^-=\mathcal D\setminus\{a\}$ and
\[
r_{\mathcal D,a}
=\mu_{\mathcal D^-}(B_a)
=1-\frac{Z_{\mathcal D}}{Z_{\mathcal D^-}}
\]
for the probability that $a$ is violated, conditioned on satisfying
all the other constraints in $\mathcal D$.

\begin{lemma}\label{lem:deletion}
For every $a\in\mathcal D\subseteq\mathcal C$,
\[
r_{\mathcal D,a}\leq cx_a.
\]
Moreover, if $\mathcal E\subseteq\mathcal D$, then
\begin{equation}\label{eq:deletion-ratio}
1\leq\frac{Z_{\mathcal E}}{Z_{\mathcal D}}
\leq\prod_{b\in\mathcal D\setminus\mathcal E}\frac1{1-cx_b}.
\end{equation}
\end{lemma}

\begin{proof}
Let $y_a=cx_a$. Since $c<1$ and $N_G(a)\subseteq N_G^2(a)$,
\[
p_a\leq cx_a\prod_{b\in N_G^2(a)}(1-x_b)
\leq y_a\prod_{b\in N_G(a)}(1-y_b).
\]
By \cite[Theorem~2.1]{haeuplerSahaSrinivasan2011}, for every
$\mathcal B\subseteq\mathcal C\setminus\{a\}$,
\[
\nu(B_a\mid\Good(\mathcal B))
\leq p_a\prod_{b\in \mathcal B\cap N_G(a)}\frac1{1-y_b}
\leq y_a.
\]
Taking $\mathcal B=\mathcal D\setminus\{a\}$ proves the first assertion. Adding a constraint $b$ to a family
$\mathcal A\setminus\{b\}$ multiplies its satisfaction probability
by $1-r_{\mathcal A,b}\geq1-cx_b$. Adding the constraints of
$\mathcal D\setminus\mathcal E$ one at a time therefore proves
the upper bound in \eqref{eq:deletion-ratio}. The lower bound follows
from $\Good(\mathcal D)\subseteq\Good(\mathcal E)$.
\end{proof}

For the rest of the paper, we fix an ordering $a_1,\ldots,a_m$ of $\mathcal C$, and use
the induced ordering on every subfamily of constraints. 
Set $\mathcal C_0=\emptyset$ and
$\mathcal C_i=\{a_1,\ldots,a_i\}$ for $1\leq i\leq m$.
The identity
\[
Z_{\mathcal C}
=\prod_{i=1}^m\frac{Z_{\mathcal C_i}}{Z_{\mathcal C_{i-1}}}
=\prod_{i=1}^m(1-r_{\mathcal C_i,a_i})
\]
shows that approximating $Z_{\mathcal C}$ reduces to estimating
the conditional probabilities $r_{\mathcal C_i,a_i}$.

For sampling, we also add the constraints one at a time. The task
at each step is to turn a sample from
$\mu_{\mathcal D^-}$ into an approximate sample from
$\mu_{\mathcal D}$. We call this an insertion of $a$.

\begin{proposition}\label{prop:insert}
For every $a\in\mathcal D\subseteq\mathcal C$ and $0<\delta\leq1$,
there is a randomized map
$\Insert(\mathcal D,a,\mathord\cdot,\delta)$ from
$\Good(\mathcal D^-)$ to $\Good(\mathcal D)$,
with expected running time polynomial in $I$ and $1/\delta$ such
that, if $\sigma\sim\mu_{\mathcal D^-}$ and
$\tau=\Insert(\mathcal D,a,\sigma,\delta)$, then
\[
\norm{\Law(\tau)-\mu_{\mathcal D}}_{\TV}\leq\delta.
\]
\end{proposition}

To see why this is useful, let $J$ be the transition kernel of
$\Insert(\mathcal D,a,\mathord\cdot,\delta)$.
For any input law $\zeta$ on
$\Good(\mathcal D^-)$, the output law
$\zeta'=\zeta J$ satisfies
\[
\begin{aligned}
\norm{\zeta'-\mu_{\mathcal D}}_{\TV}
&\leq
\norm{\zeta J-\mu_{\mathcal D^-}J}_{\TV}
+\norm{\mu_{\mathcal D^-}J-\mu_{\mathcal D}}_{\TV}\\
&\leq
\norm{\zeta-\mu_{\mathcal D^-}}_{\TV}+\delta,
\end{aligned}
\]
where we have used \cref{prop:insert} and the data processing inequality for total variation distance. Starting from a sample from $\nu$, insert $a_1,\ldots,a_m$ in
order, each with accuracy $\delta=\eps/(m+1)$. The resulting
assignment $\sigma$ satisfies every constraint in $\mathcal C$, and
iterating this inequality gives
\[
\norm{\Law(\sigma)-\mu_{\mathcal C}}_{\TV}
\leq m\delta\leq\eps.
\]
The resulting sampler has expected running time polynomial in $I$
and $1/\eps$. We prove \cref{prop:insert} in
\cref{sec:implementation}.

\subsection{The $2$-tree expansion}\label{sec:two-tree}

We first recall the $2$-tree expansion of LWYZZ
\cite{liuWangYinZhangZhou2026} for the conditional probabilities
$r_{\mathcal D,a}$. We will use the same argument to derive a
recursion for the change in distribution when a constraint is added.

For $\mathcal D\subseteq\mathcal C$, let $G_{\mathcal D}=G[\mathcal D]$. We use square brackets for closed neighborhoods. In particular, for $\mathcal T\subseteq\mathcal D$, $N_{\mathcal D}[\mathcal T]$ consists of $\mathcal T$
together with its neighbors in $G_{\mathcal D}$, and
$N_{\mathcal D}^2[\mathcal T]$ consists of all constraints at distance at most two
from $\mathcal T$ in $G_{\mathcal D}$.

\subsubsection{Conditional probabilities}

Fix $a\in\mathcal D\subseteq\mathcal C$. By inclusion--exclusion,
\[
\mathbf1_{B_a}\mathbf1_{\Good(\mathcal D^-)}
=\sum_{\substack{\mathcal B\subseteq\mathcal D\\a\in \mathcal B}}
(-1)^{|\mathcal B|-1}\prod_{b\in \mathcal B}\mathbf1_{B_b}.
\]
We will group the terms by selecting a rooted $2$-tree from each $\mathcal B$, as we will see in Theorem \ref{thm:two-tree}.
Here, a $2$-tree rooted at $a$ is a set $\mathcal T\subseteq\mathcal D$
containing $a$ that is independent in $G_{\mathcal D}$ and connected
in $G_{\mathcal D}^2 := (G_{\mathcal D})^2$. Write $\mathfrak T_{\mathcal D,a}$ for the
family of such sets.

The selection procedure that we use is due to LWYZZ
\cite[Algorithm~1]{liuWangYinZhangZhou2026}. Start with the accepted set
$\mathcal T_{\mathrm{curr}}=\{a\}$ and the rejected set $\mathcal R=\emptyset$. At each step, examine the first
constraint $b\in\mathcal D\setminus(\mathcal T_{\mathrm{curr}}\cup \mathcal R)$ whose distance
from $\mathcal T_{\mathrm{curr}}$ in $G_{\mathcal D}$ is exactly two. Add $b$ to $\mathcal T_{\mathrm{curr}}$ if
$b\in \mathcal B$, and to $\mathcal R$ otherwise. Stop when no such constraint
remains. We denote the final set $\mathcal T_{\mathrm{curr}}$ by
$\mathsf{Tree}_{\mathcal D,a}(\mathcal B)$.

For a fixed rooted instance $(\mathcal D,a)$, we suppress its dependence in all tree-indexed notation.
For $\mathcal T\in\mathfrak T_{\mathcal D,a}$, let $\mathcal R(\mathcal T)$ be the set
rejected by the procedure with input $\mathcal T$, and set
$\mathcal F(\mathcal T)=\mathcal D\setminus(\mathcal T\cup \mathcal R(\mathcal T))$. These are the constraints
that are never examined in this run. Let $\mathcal F_1(\mathcal T)$
be those constraints in $\mathcal F(\mathcal T)$ which are adjacent to a member of $\mathcal T$, and write
\[
\mathcal D_{\mathcal T}=\mathcal D\setminus N_{\mathcal D}^2[\mathcal T]
\]
for the constraints in $\mathcal{D}$ at distance greater than two from $\mathcal T$. The following properties of this selection procedure are proved in
\cite[Lemmas~3.2 and~3.4, and the proof of Lemma~4.2]
{liuWangYinZhangZhou2026}. We repeat the proof here for the reader's convenience.

\begin{lemma}\label{lem:canonical}
The possible outputs of the selection procedure are exactly the
elements of $\mathfrak T_{\mathcal D,a}$. For each such $\mathcal T$,
\[
\{\mathcal B\subseteq\mathcal D:a\in \mathcal B,
  \mathsf{Tree}_{\mathcal D,a}(\mathcal B)=\mathcal T\}
=\{\mathcal T\cup \mathcal W:\mathcal W\subseteq \mathcal F(\mathcal T)\}.
\]
Moreover,
\[
\mathcal F(\mathcal T)=\mathcal F_1(\mathcal T)\sqcup\mathcal D_{\mathcal T}.
\]
\end{lemma}

\begin{proof}
Initially, the accepted set is $\{a\}$. A constraint $b$ is added
to the current accepted set $\mathcal T_{\mathrm{curr}}$ only when
$\dist_{G_{\mathcal D}}(b,\mathcal T_{\mathrm{curr}})=2$. Thus $\mathcal T_{\mathrm{curr}}$ remains independent in
$G_{\mathcal D}$ and connected in $G_{\mathcal D}^2$, so the
output is a rooted $2$-tree.

Conversely, take $\mathcal T\in\mathfrak T_{\mathcal D,a}$ as the input.
Only constraints in $\mathcal T$ can be accepted. Suppose that the procedure
terminates with an accepted set $\mathcal T_{\mathrm{curr}}\subsetneq \mathcal T$. Since $\mathcal T$ is
connected in $G_{\mathcal D}^2$, some $b\in \mathcal T\setminus \mathcal T_{\mathrm{curr}}$ is at
distance at most two from $\mathcal T_{\mathrm{curr}}$ in $G_{\mathcal D}$. Independence
of $\mathcal T$ in $G_{\mathcal D}$ makes this distance exactly two.
The constraint $b$ cannot have been rejected, since $b$ belongs
to the input $\mathcal T$. Thus $b$ is still eligible for examination,
contradicting termination. Hence
$\mathsf{Tree}_{\mathcal D,a}(\mathcal T)=\mathcal T$.

Now consider an input $\mathcal B$ with
$\mathsf{Tree}_{\mathcal D,a}(\mathcal B)=\mathcal T$. Every examined constraint is
accepted exactly when the constraint belongs to $\mathcal T$. These are
also the decisions made with input $\mathcal T$. Since the next constraint
to be examined is determined by the accepted and rejected sets,
the runs on $\mathcal B$ and $\mathcal T$ coincide. Hence,
$\mathcal T\subseteq \mathcal B$ and $\mathcal B\cap \mathcal R(\mathcal T)=\emptyset$, so
$\mathcal B=\mathcal T\cup \mathcal W$ for some $\mathcal W\subseteq \mathcal F(\mathcal T)$.
Conversely, for any $\mathcal W\subseteq \mathcal F(\mathcal T)$, the inputs $\mathcal T$ and
$\mathcal T\cup \mathcal W$ agree on every constraint examined in the run on $\mathcal T$.
Both runs therefore make the same decisions and return $\mathcal T$.
This proves the asserted description of the sets producing $\mathcal T$.

Finally, when the procedure terminates on input $\mathcal T$, no
unexamined constraint is at distance exactly two from $\mathcal T$.
Since the accepted set is always contained in $\mathcal T$, a constraint
at distance greater than two from $\mathcal T$ can never be examined.
The unexamined constraints at distance one are precisely
$\mathcal F_1(\mathcal T)$, and those at distance greater than two are precisely
$\mathcal D_{\mathcal T}$. Therefore
$\mathcal F(\mathcal T)=\mathcal F_1(\mathcal T)\sqcup\mathcal D_{\mathcal T}$.
\end{proof}

For $\mathcal T\in\mathfrak T_{\mathcal D,a}$, let
\[
E_{\mathcal T}=\left(\bigcap_{b\in \mathcal T}B_b\right)
\cap\left(\bigcap_{b\in \mathcal F_1(\mathcal T)}B_b^c\right).
\]
Thus $E_{\mathcal T}$ is the event that every constraint in $\mathcal T$ is violated
and every constraint in $\mathcal F_1(\mathcal T)$ is satisfied. The following
identity follows from the inclusion--exclusion argument of
LWYZZ \cite{liuWangYinZhangZhou2026}.

\begin{theorem}\label{thm:two-tree}
For every $a\in\mathcal D\subseteq\mathcal C$, the following
identity holds pointwise on $\Omega$.
\begin{equation}\label{eq:pointwise-expansion}
\mathbf1_{B_a}\mathbf1_{\Good(\mathcal D^-)}
=\sum_{\mathcal T\in\mathfrak T_{\mathcal D,a}}
(-1)^{|\mathcal T|-1}\mathbf1_{E_{\mathcal T}}\mathbf1_{\Good(\mathcal D_{\mathcal T})}.
\end{equation}
\end{theorem}

\begin{proof}
Group the terms of the inclusion--exclusion formula according to
the selected tree $\mathcal T$. By \cref{lem:canonical}, the sets producing
$\mathcal T$ are exactly $\mathcal T\cup \mathcal W$ with $\mathcal W\subseteq \mathcal F(\mathcal T)$. The corresponding
terms sum to
\begin{align*}
&(-1)^{|\mathcal T|-1}\prod_{b\in \mathcal T}\mathbf1_{B_b}
\sum_{\mathcal W\subseteq \mathcal F(\mathcal T)}(-1)^{|\mathcal W|}
\prod_{b\in \mathcal W}\mathbf1_{B_b}\\
&\qquad=(-1)^{|\mathcal T|-1}\prod_{b\in \mathcal T}\mathbf1_{B_b}
\prod_{b\in \mathcal F(\mathcal T)}\mathbf1_{B_b^c}\\
&\qquad=(-1)^{|\mathcal T|-1}\mathbf1_{E_{\mathcal T}}
\mathbf1_{\Good(\mathcal D_{\mathcal T})},
\end{align*}
where the last equality uses
$\mathcal F(\mathcal T)=\mathcal F_1(\mathcal T)\sqcup\mathcal D_{\mathcal T}$.
Summing over $\mathcal T$ proves \eqref{eq:pointwise-expansion}.
\end{proof}

For a family of constraints $\mathcal A$, write
$\vbl(\mathcal A)=\bigcup_{b\in\mathcal A}\vbl(b)$.
The event $E_{\mathcal T}$ depends only on the variables in
\[
U_{\mathcal T}=\vbl(N_{\mathcal D}[\mathcal T]).
\]
We also regard $E_{\mathcal T}$ as a subset of $\Omega_{U_{\mathcal T}}$ and write
\[\theta_{\mathcal T}=\nu_{U_{\mathcal T}}(E_{\mathcal T})=\nu(E_{\mathcal T}).\]

A constraint in $\mathcal D_{\mathcal T}$ cannot involve a variable in
$U_{\mathcal T}$, since such a constraint would be at distance at most two
from $\mathcal T$. Under $\mu_{\mathcal D_{\mathcal T}}$, the coordinates in $U_{\mathcal T}$
therefore have law $\nu_{U_{\mathcal T}}$ and are independent of the
remaining coordinates. In particular,
\[
\vbl(\mathcal D_{\mathcal T})\cap U_{\mathcal T}=\emptyset,
\qquad
\mu_{\mathcal D_{\mathcal T}}(E_{\mathcal T})=\theta_{\mathcal T}.
\]
Taking expectations in \eqref{eq:pointwise-expansion} and dividing
by $Z_{\mathcal D^-}$ gives the expansion of LWYZZ
\cite[Lemma~5.1]{liuWangYinZhangZhou2026},
\begin{equation}\label{eq:scalar-recursion}
r_{\mathcal D,a}
=\sum_{\mathcal T\in\mathfrak T_{\mathcal D,a}}
(-1)^{|\mathcal T|-1}\theta_{\mathcal T}
\frac{Z_{\mathcal D_{\mathcal T}}}{Z_{\mathcal D^-}}.
\end{equation}
Since the constraints in $\mathcal T$ involve disjoint variables,
\[
\theta_{\mathcal T}
\leq\nu\left(\bigcap_{b\in \mathcal T}B_b\right)
=\prod_{b\in \mathcal T}p_b.
\]

To pass from $\mathcal D_{\mathcal T}$ to $\mathcal D^-$, list
the constraints in $N_{\mathcal D}^2[\mathcal T]\setminus\{a\}$ in the
induced order as $b_{\mathcal T,1},\ldots,b_{\mathcal T,\ell(\mathcal T)}$, and set
\[
\mathcal D_{\mathcal T,i}
=\mathcal D_{\mathcal T}\cup\{b_{\mathcal T,1},\ldots,b_{\mathcal T,i}\}
\qquad(0\leq i\leq \ell(\mathcal T)).
\]
Thus $\mathcal D_{\mathcal T,0}=\mathcal D_{\mathcal T}$ and
$\mathcal D_{\mathcal T,\ell(\mathcal T)}=\mathcal D^-$. Adding these constraints one at a time gives
\[
\frac{Z_{\mathcal D_{\mathcal T}}}{Z_{\mathcal D^-}}
=\prod_{i=1}^{\ell(\mathcal T)}
\frac{1}{1-r_{\mathcal D_{\mathcal T,i},b_{\mathcal T,i}}}.
\]
Substituting into \eqref{eq:scalar-recursion}, we obtain
\begin{equation}\label{eq:scalar-product}
r_{\mathcal D,a}
=\sum_{\mathcal T\in\mathfrak T_{\mathcal D,a}}
(-1)^{|\mathcal T|-1}\theta_{\mathcal T}
\prod_{i=1}^{\ell(\mathcal T)}
\frac{1}{1-r_{\mathcal D_{\mathcal T,i},b_{\mathcal T,i}}}.
\end{equation}
Every family $\mathcal D_{\mathcal T,i}$ is contained in
$\mathcal D^-$, so the probabilities on the
right-hand side involve strictly fewer constraints.

\subsubsection{Conditional distributions}

We next use \eqref{eq:pointwise-expansion} to derive a recursion
for $\mu_{\mathcal D^-}-\mu_{\mathcal D}$. Fix $a\in\mathcal D\subseteq\mathcal C$. For
$\mathcal T\in\mathfrak T_{\mathcal D,a}$, let
\[
\lambda_{\mathcal T}=\lambda_{\mathcal D,a}(\mathcal T)
=(-1)^{|\mathcal T|-1}\frac{Z_{\mathcal D_{\mathcal T}}}{Z_{\mathcal D}}.
\]

\begin{proposition}\label{prop:law-recursion}
For every $a\in\mathcal D\subseteq\mathcal C$,
\[
\begin{aligned}
\mu_{\mathcal D^-} -\mu_{\mathcal D}
={}&\sum_{\mathcal T\in\mathfrak T_{\mathcal D,a}}
\lambda_{\mathcal T}(\mathbf1_{E_{\mathcal T}}-\theta_{\mathcal T})\mu_{\mathcal D_{\mathcal T}}\\
&+\sum_{\mathcal T\in\mathfrak T_{\mathcal D,a}}
\lambda_{\mathcal T}\theta_{\mathcal T}
\sum_{i=1}^{\ell(\mathcal T)}
\bigl(\mu_{\mathcal D_{\mathcal T,i-1}}-\mu_{\mathcal D_{\mathcal T,i}}\bigr).
\end{aligned}
\]
\end{proposition}

\begin{proof}
Let $r:=r_{\mathcal D,a}$. Since $\mu_{\mathcal D}$ is $\mu_{\mathcal D^-}$
conditioned on $B_a^c$,
\[
\mu_{\mathcal D^-} -\mu_{\mathcal D}
=\frac{\mathbf1_{B_a}\mu_{\mathcal D^-}}{1-r}
-\frac{r}{1-r}\mu_{\mathcal D^-}.
\]
By \eqref{eq:pointwise-expansion} and the definition of $\lambda_{\mathcal T}$,
\[
\frac{\mathbf1_{B_a}\mu_{\mathcal D^-}}{1-r}
=\sum_{\mathcal T}\lambda_{\mathcal T}\mathbf1_{E_{\mathcal T}}\mu_{\mathcal D_{\mathcal T}}.
\]
Since $\mu_{\mathcal D^-}(B_a)=r$ and
$\mu_{\mathcal D_{\mathcal T}}(E_{\mathcal T})=\theta_{\mathcal T}$, we obtain
\[
\sum_{\mathcal T}\lambda_{\mathcal T}\theta_{\mathcal T}=\frac{r}{1-r}.
\]
Thus,
\begin{align*}
\mu_{\mathcal D^-} -\mu_{\mathcal D}
&=\sum_{\mathcal T}\lambda_{\mathcal T}
\bigl(\mathbf1_{E_{\mathcal T}}\mu_{\mathcal D_{\mathcal T}}-\theta_{\mathcal T}\mu_{\mathcal D^-}\bigr)\\
&= \sum_{\mathcal T}\lambda_{\mathcal T}
\bigl(\mathbf1_{E_{\mathcal T}} - \theta_{\mathcal T} \bigr)\mu_{\mathcal D_{\mathcal T}} + \sum_{\mathcal T} \lambda_{\mathcal T} \theta_{\mathcal T}\bigl( \mu_{\mathcal{D}_{\mathcal{T}}} - \mu_{\mathcal D^-}\bigr)
\end{align*}

Since $\mathcal D_{\mathcal T,0}=\mathcal D_{\mathcal T}$ and
$\mathcal D_{\mathcal T,\ell(\mathcal T)}=\mathcal D^-$, telescoping gives
\[
\mu_{\mathcal D_{\mathcal T}}-\mu_{\mathcal D^-}
=\sum_{i=1}^{\ell(\mathcal T)}
\bigl(\mu_{\mathcal D_{\mathcal T,i-1}}-\mu_{\mathcal D_{\mathcal T,i}}\bigr).
\]
Substituting this into the preceding equation proves the proposition.
\end{proof}

\section{Estimating conditional probabilities}\label{sec:estimates}

In this section, we will construct unbiased estimates of $r_{\mathcal D,a}$. We then bound
their second moments to control the number of samples needed to
approximate $Z_{\mathcal C}$, and the expected running times
for generating those samples. Our construction follows the randomized
$2$-tree expansion of LWYZZ
\cite[Sections~5.1--5.2]{liuWangYinZhangZhou2026}.

\begin{proposition}\label{prop:marginal-estimator}
For every $a\in\mathcal D\subseteq\mathcal C$, there is an estimator
$\widehat r_{\mathcal D,a}$ such that
\[
\mathbb E\widehat r_{\mathcal D,a}=r_{\mathcal D,a},
\qquad
\mathbb E\widehat r_{\mathcal D,a}^{\,2}\leq3cx_a.
\]
The expected running time to generate $\widehat r_{\mathcal D,a}$ is polynomial in $I$.
\end{proposition}

\subsection{The estimator} We use the reciprocal estimator of LWYZZ
\cite[Lemma~5.5]{liuWangYinZhangZhou2026}, with continuation
probability $1/4$. We write $\operatorname{Geom}(p)$ for the distribution on
$\{0,1,\ldots\}$ with mass $p(1-p)^k$ at $k$.

\begin{lemma}\label{lem:roulette}
Let $Y_1,Y_2,\ldots$ be independent copies of a real random variable
$Y$ with mean $r\geq0$ and second moment $v<1/4$. Independently choose
an integer $\mathbf N \sim \operatorname{Geom}(3/4)$, and set
\[
\mathsf{Recip}(Y)
=\sum_{k=0}^{\mathbf N}4^k\prod_{j=1}^kY_j.
\]
Then
\[
\mathbb E\mathsf{Recip}(Y)=\frac1{1-r},
\qquad
\mathbb E\mathsf{Recip}(Y)^2
=\frac{1+r}{(1-r)(1-4v)},
\qquad
\mathbb E\mathbf N=\frac13.
\]
\end{lemma}

We define $\widehat r_{\mathcal D,a}$ by the recursive procedure
in \cref{alg:estimate-marginal}. Since
$\mathcal D_{\mathcal T,i}\subseteq\mathcal D^-$, every recursive
call involves fewer constraints.

For each $b\in\mathcal D$, independently draw
$\sigma_b\sim\nu_{\vbl(b)}$, and call $b$ active if $\sigma_b$
violates $b$. For each $\mathcal T\in\mathfrak T_{\mathcal D,a}$, the sets
$\vbl(t)$, $t\in \mathcal T$, are disjoint, so the assignments $\sigma_t$
combine to give an assignment on $\vbl(\mathcal T)$. Complete this
assignment to $U_{\mathcal T}$ by sampling the remaining coordinates from
their product distribution, using fresh randomness for each $\mathcal T$.
Denote the resulting assignment by $\sigma_{\mathcal T}$. For every fixed
$\mathcal T$, we have $\sigma_{\mathcal T}\sim\nu_{U_{\mathcal T}}$, and hence
\[
\mathbb P(\sigma_{\mathcal T}\in E_{\mathcal T})=\theta_{\mathcal T}.
\]
Thus $\mathbf1_{E_{\mathcal T}}(\sigma_{\mathcal T})$ is an unbiased estimate of
$\theta_{\mathcal T}$.

If any constraint in $\mathcal T$ is inactive, then $\sigma_{\mathcal T}\notin E_{\mathcal T}$,
regardless of the completion. We therefore only generate
$\sigma_{\mathcal T}$ for trees whose constraints are all active. In
particular, if $a$ is inactive, the estimator is zero.

The remaining factors in \eqref{eq:scalar-product} are
reciprocals of the form $(1-r_{\mathcal E,b})^{-1}$, with
$\mathcal E\subsetneq\mathcal D$. For each such factor, use
independent recursive calls to generate copies of
$\widehat r_{\mathcal E,b}$, and apply \cref{lem:roulette}. Write
\[
\widehat R_{\mathcal E,b}
=\mathsf{Recip}(\widehat r_{\mathcal E,b}).
\]
For each pair $(\mathcal T,i)$, let $\widehat R_{\mathcal T,i}$ be a separate copy of
$\widehat R_{\mathcal D_{\mathcal T,i},b_{\mathcal T,i}}$. These copies are
independent of one another and of all the local assignments.
Define
\[
\widehat r_{\mathcal D,a}
=\sum_{\mathcal T\in\mathfrak T_{\mathcal D,a}}
(-1)^{|\mathcal T|-1}\mathbf1_{E_{\mathcal T}}(\sigma_{\mathcal T})
\prod_{i=1}^{\ell(\mathcal T)}
\widehat R_{\mathcal T,i}.
\]
In computing this sum, we generate the reciprocal estimates only
after the test $\sigma_{\mathcal T}\in E_{\mathcal T}$ succeeds. When $\mathcal D=\{a\}$, the only tree is $\{a\}$ and the product
is empty. The estimator is therefore just the indicator that
$\sigma_a$ violates $a$, which gives the base case of the
recursion. We discuss the enumeration of active trees in the
running-time analysis below.

\begin{algorithm}[H]
\caption{$\EstimateMarginal(\mathcal D,a)$}\label{alg:estimate-marginal}
\KwIn{A constraint family $\mathcal D$ and a root $a\in\mathcal D$.}
\KwOut{One realization of $\widehat r_{\mathcal D,a}$.}

Independently draw $\sigma_b\sim\nu_{\vbl(b)}$ for every
$b\in\mathcal D$\;
$\mathcal A\gets\{b\in\mathcal D:\sigma_b\text{ violates }b\}$\;
\lIf{$a\notin \mathcal A$}{\Return $0$}

$\widehat r\gets0$\;
\ForEach{rooted $2$-tree $\mathcal T\in\mathfrak T_{\mathcal D,a}$
with $\mathcal T\subseteq \mathcal A$}{
  Combine the assignments $\sigma_t$, $t\in \mathcal T$, and complete
  them to $\sigma_{\mathcal T}\in\Omega_{U_{\mathcal T}}$ with fresh product samples\;
  \If{$\sigma_{\mathcal T}\in E_{\mathcal T}$}{
    $W\gets1$\;
    \For{$i=1,\ldots,\ell(\mathcal T)$}{
      Independently draw $\mathbf N\sim\operatorname{Geom}(3/4)$\;
      $P\gets1$; $\widehat{R}_i\gets1$\;
      \For{$k=1,\ldots,\mathbf N$}{
        $Y\gets\EstimateMarginal(\mathcal D_{\mathcal T,i},b_{\mathcal T,i})$,
        using fresh randomness\;
        $P\gets (4Y) \cdot P$; $\widehat{R}_i\gets \widehat{R}_i+P$\;
      }
      $W\gets W\,\widehat{R}_i$\;
    }
    $\widehat r\gets \widehat r+(-1)^{|\mathcal T|-1}W$\;
  }
}
\Return $\widehat r$\;
\end{algorithm}
\subsection{Mean and second moment}

We now prove the moment bounds for $\widehat r_{\mathcal D,a}$
in \cref{prop:marginal-estimator}. These bounds ensure that
average of independent copies concentrates around
$r_{\mathcal D,a}$. Because the estimator is recursive, we also
need bounds on the moments of the reciprocal estimates.

Let
\[
h_b=\frac1{1-cx_b},
\qquad
j_b=e^{16cx_b}.
\]
We prove the moment bounds in \cref{prop:marginal-estimator} together with
\[
\mathbb E\widehat R_{\mathcal D,a}
=\frac1{1-r_{\mathcal D,a}}\leq h_a,
\qquad
\mathbb E\widehat R_{\mathcal D,a}^{\,2}\leq j_a.
\]

Squaring the sum defining $\widehat r_{\mathcal D,a}$ gives
terms indexed by pairs of rooted $2$-trees $\mathcal T, \mathcal T'$. The union $\mathcal T\cup \mathcal T'$ is connected in $G_{\mathcal D}^2$, as both
trees contain $a$. To bound the sum over these unions, we use
the following consequence of the branching-process argument of
Moser and Tardos \cite[Section~3]{moserTardos2010}.

\begin{lemma}\label{lem:connected-sets}
Let $H=(V,E)$ be a finite simple graph and let $0<\alpha<1$.
Suppose that $z\in(0,1)^V$ and $w\in[0,\infty)^V$ satisfy
\[
w_v\leq\alpha z_v\prod_{u\in N_H(v)}(1-z_u)
\qquad(v\in V).
\]
Then, for every $a\in V$,
\[
\sum_{\substack{S\subseteq V,\ a\in S\\
                 H[S]\text{ is connected}}}
\prod_{v\in S}w_v
\leq\alpha\frac{z_a}{1-z_a}.
\]
\end{lemma}

\begin{proof}
For each connected set $S$ containing $a$, fix a spanning tree
of $H[S]$ rooted at $a$. Consider a branching process starting
with a vertex of type $a$. A vertex of type $v$ independently
has a child of each type $u\in N_H[v]$ with probability $z_u$.
The probability that the process produces the chosen tree on
$S$ is
\[
\frac{1-z_a}{z_a}
\prod_{v\in S}\left(z_v\prod_{u\in N_H(v)}(1-z_u)\right).
\]
Indeed, multiplying the probabilities of including each child
and excluding all other possible children gives this expression.
The output events for different $S$ are disjoint, so their
probabilities sum to at most one. The hypothesis on $w_v$ and
$\alpha^{|S|}\leq\alpha$ give the result.
\end{proof}

We will also use the following pair bound. If $\mathfrak S$ is any family of connected vertex sets of a graph $H=(V,E)$, each containing $a$, and $w_v\geq0$, then
\[
\sum_{S,S'\in\mathfrak S}\prod_{v\in S\cup S'}w_v
\leq\frac13\sum_{\substack{B\subseteq V,\ a\in B\\H[B]\text{ is connected}}}
\prod_{v\in B}(3w_v).
\]
Indeed, each union $B=S\cup S'$ is connected, and at most $3^{|B|-1}$ ordered pairs have this union: $a$ belongs to both sets, and every other vertex belongs to the first, the second, or both \cite[Lemma~5.6]{liuWangYinZhangZhou2026}.

Define the global weights
\[
\rho_t=p_t\prod_{b\in N_G^2[t]}j_b
\qquad(t\in\mathcal C).
\]

\begin{proof}[Proof of the moment bounds in \cref{prop:marginal-estimator}]
We prove the moment bounds for $\widehat r_{\mathcal D,a}$ and
$\widehat R_{\mathcal D,a}$ simultaneously by induction on
$|\mathcal D|$.
Unbiasedness follows from the induction hypothesis,
independence, and \eqref{eq:scalar-recursion}, as in
\cite[Lemma~5.6]{liuWangYinZhangZhou2026}.

For $\mathcal T\in\mathfrak T_{\mathcal D,a}$, let
\[
{R}_{\mathcal T}=\prod_{i=1}^{\ell(\mathcal T)}
\widehat R_{\mathcal T,i},
\qquad
Y_{\mathcal T}=(-1)^{|\mathcal T|-1}\mathbf1_{E_{\mathcal T}}(\sigma_{\mathcal T}){R}_{\mathcal T}.
\]
Then $\widehat r_{\mathcal D,a}
=\sum_{\mathcal T\in\mathfrak T_{\mathcal D,a}}Y_{\mathcal T}$.

By the induction hypothesis,
\[
\mathbb E{R}_{\mathcal T}^2
\leq\prod_{b\in N_{\mathcal D}^2[\mathcal T]\setminus\{a\}}j_b.
\]
The events $\sigma_{\mathcal T}\in E_{\mathcal T}$ and $\sigma_{\mathcal T'}\in E_{\mathcal T'}$ can
both occur only if every constraint in $\mathcal T\cup \mathcal T'$ is active.
Since the assignments $\sigma_b$ are independent, the joint
probability is at most $\prod_{t\in \mathcal T\cup \mathcal T'}p_t$.
For fixed $\mathcal T,\mathcal T'$, the pair $({R}_{\mathcal T},{R}_{\mathcal T'})$ is independent of
$(\sigma_{\mathcal T},\sigma_{\mathcal T'})$. Hence
\[
\begin{aligned}
|\mathbb E Y_{\mathcal T}Y_{\mathcal T'}|
&=\mathbb P(\sigma_{\mathcal T}\in E_{\mathcal T},\ \sigma_{\mathcal T'}\in E_{\mathcal T'})
  |\mathbb E {R}_{\mathcal T}{R}_{\mathcal T'}|\\
&\leq
\left(\prod_{t\in \mathcal T\cup \mathcal T'}p_t\right)
\sqrt{\mathbb E {R}_{\mathcal T}^2\,\mathbb E {R}_{\mathcal T'}^2}\\
&\leq
\left(\prod_{t\in \mathcal T\cup \mathcal T'}p_t\right)
\prod_{b\in N_{\mathcal D}^2[\mathcal T\cup \mathcal T']}j_b\\
&\leq\prod_{t\in \mathcal T\cup \mathcal T'}\rho_t.
\end{aligned}
\]
Here we used Cauchy--Schwarz and the preceding bounds on
$\mathbb E {R}_{\mathcal T}^2$ and $\mathbb E {R}_{\mathcal T'}^2$. The last inequality
follows from the definition of $\rho_t$ and $j_b\geq1$.

The pair bound following \cref{lem:connected-sets}, applied in $G_{\mathcal D}^2$, gives
\[
\mathbb E\widehat r_{\mathcal D,a}^{\,2}
\leq\frac13
\sum_{\substack{\mathcal B\ni a\\G_{\mathcal D}^2[\mathcal B]\text{ is connected}}}
\prod_{t\in \mathcal B}(3\rho_t).
\]

To apply \cref{lem:connected-sets}, note that
$(1-x)e^{16cx}\leq1-x/2$ for $0<x<1$, since $16c<1/2$.
Using \eqref{eq:full-square} and $e^{16c}<4/3$, we obtain
\begin{equation}\label{eq:common-weight}
\begin{aligned}
\rho_t
&\leq cx_t e^{16cx_t}
\prod_{b\in N_G^2(t)}
(1-x_b)e^{16cx_b}\\
&\leq\frac43cx_t
\prod_{b\in N_{\mathcal D}^2(t)}(1-x_b/2).
\end{aligned}
\end{equation}
Applying \cref{lem:connected-sets} with $H=G_{\mathcal D}^2$,
$w_t=3\rho_t$, $z_t=x_t/2$, and $\alpha=8c$ gives
\begin{equation}\label{eq:pair-sum}
\frac13
\sum_{\substack{\mathcal B\ni a\\G_{\mathcal D}^2[\mathcal B]\text{ is connected}}}
\prod_{t\in \mathcal B}(3\rho_t)
\leq\frac{8c}{3}\frac{x_a}{2-x_a}
\leq\frac{8c}{3}x_a.
\end{equation}
This proves the required second-moment bound.

Since $\mathbb E\widehat r_{\mathcal D,a}^{\,2}
\leq3cx_a<1/4$, we may now apply \cref{lem:roulette} to
$\widehat r_{\mathcal D,a}$. Using $r_{\mathcal D,a}\leq cx_a$,
we obtain
\[
\begin{aligned}
\mathbb E\widehat R_{\mathcal D,a}
&=\frac1{1-r_{\mathcal D,a}}\leq h_a,\\
\mathbb E\widehat R_{\mathcal D,a}^{\,2}
&\leq\frac{1+cx_a}{(1-cx_a)(1-12cx_a)}
\leq e^{16cx_a}=j_a.
\end{aligned}
\]
The last inequality follows from
\[
\log\frac{1+t}{(1-t)(1-12t)}
\leq t+\frac{t}{1-t}+\frac{12t}{1-12t}
\leq16t
\qquad(0\leq t\leq c).
\]
This proves the reciprocal-estimator bounds and completes the
induction. The same estimate also gives $1\leq h_b^2\leq j_b$.
\end{proof}

\subsection{Running time}

Following \cite[Remark~5.2]{liuWangYinZhangZhou2026}, we enumerate
without repetition the connected sets of active constraints
containing $a$ in $G_{\mathcal D}^2$, using polynomial time
in $I$ per set. Retaining those that are independent in
$G_{\mathcal D}$ gives precisely the active rooted $2$-trees.
Let $A$ denote the number of connected sets enumerated.

For each active tree $\mathcal T$, completing $\sigma_{\mathcal T}$ and testing
$E_{\mathcal T}$ take polynomial time in $I$. If the test succeeds, we
compute at most $m$ reciprocal estimates. By \cref{lem:roulette},
each requires an expected constant amount of work apart from
its recursive calls. Thus the expected work within a call,
excluding recursive calls, is bounded by a polynomial in $I$
times $1+\mathbb EA$.

To bound the recursive work, we will prove that computing
$\widehat r_{\mathcal D,a}$ makes at most $1+mx_a$ calls in
expectation, including the initial call. The recursive calls
associated with a tree $\mathcal T$ have roots
$b_{\mathcal T,1},\ldots,b_{\mathcal T,\ell(\mathcal T)}$, so summing their inductive bounds
introduces
\[
\ell_x(\mathcal T)=\sum_{i=1}^{\ell(\mathcal T)}x_{b_{\mathcal T,i}}
=\sum_{b\in N_{\mathcal D}^2[\mathcal T]\setminus\{a\}}x_b.
\]
The following lemma bounds $\mathbb EA$ and the weighted tree
sums needed for this induction. The second-moment bound on $A$ in this lemma will also be used in the sampling analysis.

\begin{lemma}\label{lem:tree-bounds}
For every $a\in\mathcal D\subseteq\mathcal C$, the number $A$
of active connected sets containing $a$ in $G_{\mathcal D}^2$
satisfies
\[
\mathbb EA\leq\mathbb EA^2\leq2cx_a.
\]
Moreover,
\begin{equation}\label{eq:tree-weight-sum}
\begin{aligned}
\sum_{\mathcal T\in\mathfrak T_{\mathcal D,a}}\prod_{t\in \mathcal T}\rho_t
&\leq4cx_a,\\
\sum_{\mathcal T\in\mathfrak T_{\mathcal D,a}}
\ell_x(\mathcal T)\prod_{t\in \mathcal T}\rho_t
&<10cx_a.
\end{aligned}
\end{equation}
\end{lemma}

\begin{proof}
Independence of the activation events and the pair bound following
\cref{lem:connected-sets} give
\[
\mathbb EA^2
\leq\frac13
\sum_{\substack{\mathcal B\subseteq\mathcal D,\ a\in \mathcal B\\
                G_{\mathcal D}^2[\mathcal B]\text{ is connected}}}
\prod_{t\in \mathcal B}(3p_t)
\leq2cx_a.
\]
The last inequality follows from \cref{lem:connected-sets}
with $z_t=x_t/2$ and $\alpha=6c$, since
\[
p_t\leq cx_t
\prod_{b\in N_{\mathcal D}^2(t)}(1-x_b/2).
\]
Also, $A\leq A^2$ because $A$ is a nonnegative integer.

Every rooted $2$-tree is connected in $G_{\mathcal D}^2$.
By \eqref{eq:common-weight}, we may apply
\cref{lem:connected-sets} with $w_t=\rho_t$, $z_t=x_t/2$,
and $\alpha=8c/3$, obtaining
\[
\sum_{\mathcal T}\prod_{t\in \mathcal T}\rho_t
\leq\frac{8c}{3}\frac{x_a}{2-x_a}
\leq4cx_a.
\]

For the sum containing $\ell_x(\mathcal T)$, set
\[
\widetilde\rho_t
=\rho_t\exp\left(
\frac14\sum_{b\in N_G^2[t]}x_b
\right).
\]
Since $16c+1/4<1/2$ and $e^{16c+1/4}<2$, the calculation
in \eqref{eq:common-weight} also gives
\[
\widetilde\rho_t
\leq2cx_t
\prod_{b\in N_{\mathcal D}^2(t)}(1-x_b/2).
\]
Hence, applying \cref{lem:connected-sets} with $z_t=x_t/2$ and $\alpha=4c$,
we get $\sum_{\mathcal T}\prod_{t\in \mathcal T}\widetilde\rho_t\leq4cx_a$.
Now
\[
\ell_x(\mathcal T)\leq
\sum_{t\in \mathcal T}\sum_{b\in N_G^2[t]}x_b.
\]
Using $s e^{-s/4}\leq4/\mathrm e$ for $s\geq0$, we obtain
\[
\sum_{\mathcal T} \ell_x(\mathcal T)\prod_{t\in \mathcal T}\rho_t
\leq\frac4{\mathrm e}
\sum_{\mathcal T}\prod_{t\in \mathcal T}\widetilde\rho_t
\leq\frac{16c}{\mathrm e}x_a
<10cx_a.
\]
\end{proof}

We can now bound the total number of recursive calls.

\begin{lemma}\label{lem:recursive-work}
For every $a\in\mathcal D\subseteq\mathcal C$, the expected
total number of calls to the marginal estimator in computing
$\widehat r_{\mathcal D,a}$, including the initial call,
is at most $1+mx_a$.
\end{lemma}

\begin{proof}
We induct on $|\mathcal D|$. If $\mathcal D=\{a\}$, there
are no recursive calls.

Fix $\mathcal T\in\mathfrak T_{\mathcal D,a}$. Each reciprocal estimate
for this tree uses an expected $1/3$ recursive calls.
The call rooted at $b_{\mathcal T,i}$ involves a proper subfamily of
$\mathcal D$, so the induction hypothesis bounds its expected
number of calls, including all descendants, by
$1+mx_{b_{\mathcal T,i}}$. Computing all the reciprocal estimates for
$\mathcal T$ therefore requires at most
\[
\frac13\sum_{i=1}^{\ell(\mathcal T)}(1+mx_{b_{\mathcal T,i}})
=\frac13\bigl(\ell(\mathcal T)+m\ell_x(\mathcal T)\bigr)
\]
calls in expectation.

This computation is performed only if $\sigma_{\mathcal T}\in E_{\mathcal T}$,
an event of probability
$\theta_{\mathcal T}\leq\prod_{t\in \mathcal T}p_t$, and uses fresh randomness.
Thus the expected total number of calls is at most
\[
\begin{aligned}
1+\frac13\sum_{\mathcal T}
\left(\prod_{t\in \mathcal T}p_t\right)
\bigl(\ell(\mathcal T)+m\ell_x(\mathcal T)\bigr)
&\leq
1+\frac m3\sum_{\mathcal T}(1+\ell_x(\mathcal T))\prod_{t\in \mathcal T}\rho_t\\
&\leq1+\frac{14c}{3}mx_a\\
&\leq1+mx_a.
\end{aligned}
\]
Here we used $\ell(\mathcal T)\leq m$, $p_t\leq \rho_t$, and
\cref{lem:tree-bounds}.
\end{proof}

By \cref{lem:tree-bounds} and the enumeration bound above,
the expected work within each call, excluding recursive calls,
is polynomial in $I$. Since each call uses fresh randomness,
this bound holds conditional on the history before the call
begins. By \cref{lem:recursive-work}, the expected total number
of calls is at most $1+mx_a\leq1+m$. Summing the work over all
calls therefore gives polynomial expected running time,
proving the running-time assertion of
\cref{prop:marginal-estimator}.

\section{Approximate counting}\label{sec:counting}

To approximate
$Z_{\mathcal C}=\prod_{i=1}^m(1-r_{\mathcal C_i,a_i})$,
we estimate each $r_{\mathcal C_i,a_i}$ by averaging independent
copies of $\widehat r_{\mathcal C_i,a_i}$ and substitute these
averages into the product, as follows.

\begin{algorithm}[H]
\caption{$\Count(\mathcal C,\eps)$}\label{alg:count}
\KwIn{An ordered family $\mathcal C=\{a_1,\ldots,a_m\}$ and accuracy $0<\eps\leq1$.}
\KwOut{An estimate of $Z_{\mathcal C}$ in $[0,1]$.}
$n\gets\left\lceil16(m+1)/\eps^2\right\rceil$; $\widetilde Z\gets1$\;
\For{$i=1,\ldots,m$}{
  $\mathcal C_i\gets\{a_1,\ldots,a_i\}$; $s\gets0$\;
  \For{$j=1,\ldots,n$}{
    $s\gets s+\EstimateMarginal(\mathcal C_i,a_i)$, using fresh randomness\;
  }
  $\overline r_i\gets s/n$; $\widetilde Z\gets\widetilde Z(1-\overline r_i)$\;
}
\Return $\min\{1,\max\{0,\widetilde Z\}\}$\;
\end{algorithm}

\begin{proof}[Proof of the counting assertion of \cref{thm:main}]
As in \cite[Section~5.3]{liuWangYinZhangZhou2026}, independence
and \cref{prop:marginal-estimator} give
$\mathbb E\widetilde Z=Z_{\mathcal C}$ and
\[
\frac{\operatorname{Var}(\widetilde Z)}{Z_{\mathcal C}^2}
\leq\exp\left(
\frac{3c}{n(1-c)^2}\sum_{i=1}^m x_{a_i}
\right)-1
\leq e^{\eps^2/16}-1<\frac{\eps^2}{8}.
\]
Here we used $3c/(1-c)^2<1$ and the choice of $n$.
Chebyshev's inequality gives
\[
\mathbb P\bigl(|\widetilde Z-Z_{\mathcal C}|>
\eps Z_{\mathcal C}\bigr)<\frac18.
\]
Clipping to $[0,1]$ cannot increase the error.

By \cref{prop:marginal-estimator}, the expected running time is
polynomial in $I$ and $1/\eps$. By Markov's inequality, we can choose a polynomial time limit
that the algorithm exceeds with probability at most $1/8$.
If this happens, we return $0$. The resulting algorithm has
worst-case polynomial running time and returns an estimate
with relative error at most $\eps$ with probability at least
$3/4$. Taking the median of
$O(\log(2/\delta))$ independent repetitions gives success
probability at least $1-\delta$ and running time polynomial in
$I$, $1/\eps$, and $\log(1/\delta)$.
\end{proof}

\section{The insertion chain}\label{sec:insertion}

We now construct the chain used in the insertion map of
\cref{prop:insert}. Fix $a\in\mathcal D\subseteq\mathcal C$,
and write $\mathcal D^-=\mathcal D\setminus\{a\}$. Our aim is
to turn a sample from $\mu_{\mathcal D^-}$ into an approximate
sample from $\mu_{\mathcal D}$.

While the construction of the chain is quite involved, it is crucially guided by the identity in \cref{prop:law-recursion}. 
In particular, its first local term is implemented by resampling $U_{\mathcal T}$. Its remaining difference of distributions over smaller instances of the problem is implemented by recursing the construction through a corridor that links together the chain construction over smaller instances.

\subsection{The state space}

Write $\mathsf S(\mathcal D,a)$ for the state space. We define
the state spaces by induction on $|\mathcal D|$.

There is an \emph{input state} $\mathsf{I}(\sigma)$ for each
$\sigma\in\Good(\mathcal D^-)$ and an \emph{output state}
$\mathsf{O}(\sigma)$ for each $\sigma\in\Good(\mathcal D)$. These are
distinct states even when their assignments agree.

For each $\mathcal T\in\mathfrak T_{\mathcal D,a}$, set
\[
E_{\mathcal T}^{\star}
=E_{\mathcal T}\cap\Good(\mathcal D_{\mathcal T}),
\]
and introduce a \emph{tree state} $\mathsf{T}_{\mathcal T}(\sigma)$
for each $\sigma\in E_{\mathcal T}^{\star}$.

For each $\mathcal T\in\mathfrak T_{\mathcal D,a}$ and
$z\in E_{\mathcal T}\subseteq\Omega_{U_{\mathcal T}}$,
write $e=(\mathcal T,z)$ and introduce the \emph{layers}
\[
\mathsf{L}_i^e
=
\{\mathsf{L}_i^e(\sigma):
  \sigma\in\Good(\mathcal D_{\mathcal T,i})\}
\qquad(0\leq i\leq\ell(\mathcal T)).
\]
The first layer contains assignments satisfying
$\mathcal D_{\mathcal T}$, and the last contains assignments
satisfying $\mathcal D^-$. Each intervening layer imposes
one additional constraint.

In a state $\mathsf{L}_i^e(\sigma)$, the full assignment
$\sigma$ and the auxiliary local assignment $z$ are separate
parts of the state. For a fixed $\mathcal T$, write
$\sigma=(y,u)$, where $u = \sigma|_{U_\mathcal{T}}$. There is no requirement that $u=z$. The condition
$z\in E_{\mathcal T}$ applies to the label $z$, whereas
$\sigma$ is required to satisfy $\mathcal D_{\mathcal T,i}$.

We complete the state space by adding \emph{internal states} from
smaller chains. For $b\in\mathcal E\subsetneq\mathcal D$, let
\[
\mathsf S^\circ(\mathcal E,b)
=
\mathsf S(\mathcal E,b)\setminus
\left(
\{\mathsf{I}(\sigma):
  \sigma\in\Good(\mathcal E\setminus\{b\})\}
\sqcup
\{\mathsf{O}(\sigma):
  \sigma\in\Good(\mathcal E)\}
\right).
\]
Thus $\mathsf S^\circ(\mathcal E,b)$ consists of all states
of the smaller chain except its input and output states.

For each $e=(\mathcal T,z)$ as above and
$1\leq i\leq\ell(\mathcal T)$, introduce a distinct state
\[
(e,i,X)
\qquad
\left(
X\in
\mathsf S^\circ(\mathcal D_{\mathcal T,i},b_{\mathcal T,i})
\right).
\]
These smaller state spaces are already defined, since
$\mathcal D_{\mathcal T,i}\subseteq\mathcal D^-\subsetneq\mathcal D$.
The state space $\mathsf S(\mathcal D,a)$ consists of the input,
output, tree, and layer states defined above, together with
the states $(e,i,X)$ representing the internal states
of the smaller copies. States with different labels are
distinct.

For each $e$ and $i$, define the embedding
\[
\psi_{e,i}:
\mathsf S(\mathcal D_{\mathcal T,i},b_{\mathcal T,i})
\longrightarrow \mathsf S(\mathcal D,a)
\]
by
\[
\psi_{e,i}(X)=
\begin{cases}
\mathsf{L}_{i-1}^e(\sigma),
  & X=\mathsf{I}(\sigma),\\
\mathsf{L}_i^e(\sigma),
  & X=\mathsf{O}(\sigma),\\
(e,i,X),
  & X\in
    \mathsf S^\circ(\mathcal D_{\mathcal T,i},b_{\mathcal T,i}).
\end{cases}
\]
This identifies the input and output states of the smaller
copy with layers $i-1$ and $i$, respectively. The assignment
spaces agree because
\[
\mathcal D_{\mathcal T,i}\setminus\{b_{\mathcal T,i}\}
=\mathcal D_{\mathcal T,i-1}.
\]
Only the internal states of the copy are added as new states.

For fixed $e=(\mathcal T,z)$, the layers and the copies joining them
form a \emph{corridor}. Adjacent copies share one layer and have
no other states in common. The construction is illustrated in
\cref{fig:state-space} for a corridor with $\ell(\mathcal T)=3$.
If $\ell(\mathcal T)=0$, the corridor consists only of
$\mathsf{L}_0^e$. In particular, when $|\mathcal D|=1$,
there are no smaller copies, which starts the induction.

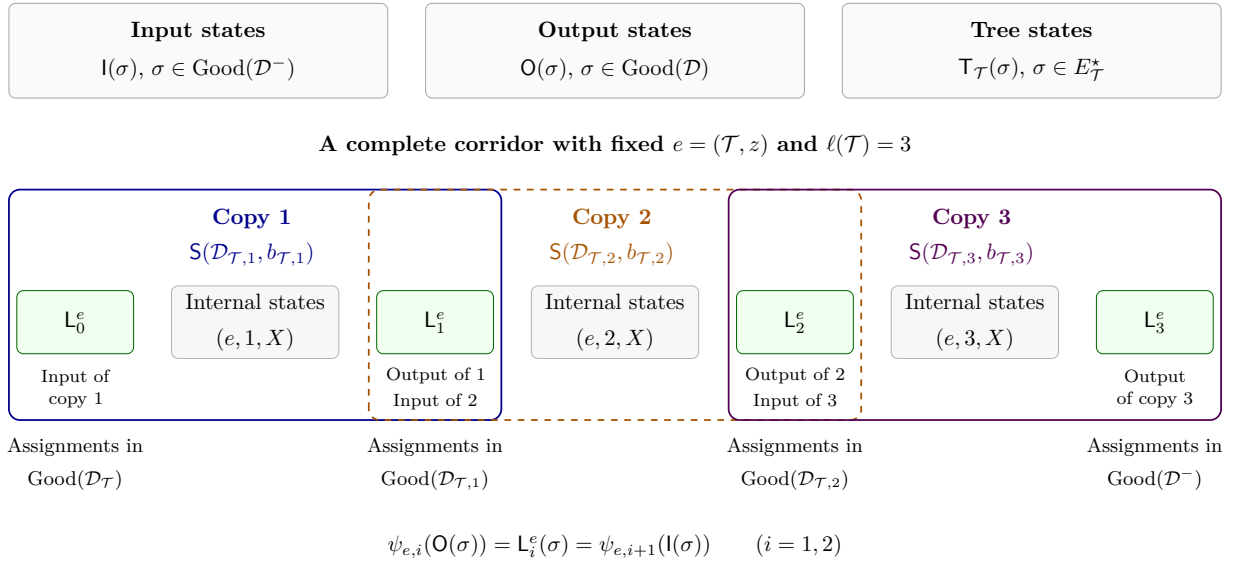
\begin{figure}[htbp]
\centering
\resizebox{\textwidth}{!}{
\begin{tikzpicture}[
  x=1cm,y=1cm,
  every node/.style={font=\small},
  group/.style={
    draw=black!35,rounded corners=3pt,fill=black!2,
    minimum width=6cm,minimum height=1.5cm,align=center
  },
  layer/.style={
    draw=green!40!black,fill=green!6,rounded corners=3pt,
    minimum width=1.85cm,minimum height=1cm,align=center
  },
  internal/.style={
    draw=black!30,fill=black!3,rounded corners=3pt,
    minimum width=2.65cm,minimum height=1cm,align=center
  },
  role/.style={
    font=\scriptsize,align=center,text width=1.9cm
  }
]

\node[group] at (3,7)
  {\textbf{Input states}\\[5pt]
   $\mathsf{I}(\sigma)$, $\sigma\in\Good(\mathcal D^-)$};
\node[group] at (9.6,7)
  {\textbf{Output states}\\[5pt]
   $\mathsf{O}(\sigma)$, $\sigma\in\Good(\mathcal D)$};
\node[group] at (16.2,7)
  {\textbf{Tree states}\\[5pt]
   $\mathsf{T}_{\mathcal T}(\sigma)$, $\sigma\in E_{\mathcal T}^{\star}$};

\node[font=\small\bfseries] at (9.6,5.5)
  {A complete corridor with fixed $e=(\mathcal T,z)$
   and $\ell(\mathcal T)=3$};

\draw[blue!55!black,line width=.8pt,rounded corners=5pt]
  (0,1.15) rectangle (7.8,4.8);
\draw[orange!65!black,line width=.8pt,dashed,rounded corners=5pt]
  (5.7,1.15) rectangle (13.5,4.8);
\draw[violet!65!black,line width=.8pt,rounded corners=5pt]
  (11.4,1.15) rectangle (19.2,4.8);

\node[align=center,text=blue!55!black] at (3.85,4.05)
  {\textbf{Copy 1}\\[4pt]
   $\mathsf S(\mathcal D_{\mathcal T,1},b_{\mathcal T,1})$};
\node[align=center,text=orange!65!black] at (9.55,4.05)
  {\textbf{Copy 2}\\[4pt]
   $\mathsf S(\mathcal D_{\mathcal T,2},b_{\mathcal T,2})$};
\node[align=center,text=violet!65!black] at (15.25,4.05)
  {\textbf{Copy 3}\\[4pt]
   $\mathsf S(\mathcal D_{\mathcal T,3},b_{\mathcal T,3})$};

\node[layer] at (1.05,2.7) {$\mathsf{L}_0^e$};
\node[layer] at (6.75,2.7) {$\mathsf{L}_1^e$};
\node[layer] at (12.45,2.7) {$\mathsf{L}_2^e$};
\node[layer] at (18.15,2.7) {$\mathsf{L}_3^e$};

\node[internal] at (3.9,2.7)
  {Internal states\\[4pt]$(e,1,X)$};
\node[internal] at (9.6,2.7)
  {Internal states\\[4pt]$(e,2,X)$};
\node[internal] at (15.3,2.7)
  {Internal states\\[4pt]$(e,3,X)$};

\node[role] at (1.05,1.65) {Input of copy 1};
\node[role] at (6.75,1.65)
  {Output of 1\\[2pt]Input of 2};
\node[role] at (12.45,1.65)
  {Output of 2\\[2pt]Input of 3};
\node[role] at (18.15,1.65) {Output of copy 3};

\node[align=center,font=\footnotesize] at (1.05,.45)
  {Assignments in\\[4pt]$\Good(\mathcal D_{\mathcal T})$};
\node[align=center,font=\footnotesize] at (6.75,.45)
  {Assignments in\\[4pt]$\Good(\mathcal D_{\mathcal T,1})$};
\node[align=center,font=\footnotesize] at (12.45,.45)
  {Assignments in\\[4pt]$\Good(\mathcal D_{\mathcal T,2})$};
\node[align=center,font=\footnotesize] at (18.15,.45)
  {Assignments in\\[4pt]$\Good(\mathcal D^-)$};

\node[align=center] at (9.6,-.85)
  {$\psi_{e,i}(\mathsf{O}(\sigma))
    =\mathsf{L}_i^e(\sigma)
    =\psi_{e,i+1}(\mathsf{I}(\sigma))
    \qquad(i=1,2)$};

\end{tikzpicture}
}
\caption{The input, output, and tree states, together with one
complete corridor for $\ell(\mathcal T)=3$.
Each outlined region represents an embedded smaller state space.
Consecutive copies share exactly one layer, which is the output
of one copy and the input of the next.
The label $z$ is fixed throughout the corridor and need not equal
the current assignment on $U_{\mathcal T}$.
No transitions are shown.}
\label{fig:state-space}
\end{figure}

Every state determines a current assignment. At
$\mathsf{I}(\sigma)$, $\mathsf{O}(\sigma)$, $\mathsf{T}_{\mathcal T}(\sigma)$, or
$\mathsf{L}_i^e(\sigma)$, it is $\sigma$. At
$(e,i,X)$, it is the current assignment at $X$
in the smaller chain.

Finally, define the input and target distributions on
$\mathsf S(\mathcal D,a)$ by
\[
\mu^-(\mathsf{I}(\sigma))=\mu_{\mathcal D^-}(\sigma),
\qquad
\mu^+(\mathsf{O}(\sigma))=\mu_{\mathcal D}(\sigma),
\]
with zero mass on all other states. Thus $\mu^-$ is the law
of $\mathsf{I}(\sigma)$ when $\sigma\sim\mu_{\mathcal D^-}$, and
$\mu^+$ is the law of $\mathsf{O}(\sigma)$ when
$\sigma\sim\mu_{\mathcal D}$.

\subsection{Transitions}

We first define an auxiliary transition kernel $P_{\mathcal D,a}$
on $\mathsf S(\mathcal D,a)$, abbreviated to $P$ for this fixed
instance. Every smaller copy uses its unmodified auxiliary kernel.
We then modify only the transitions from the outer output states
to obtain $K^\beta_{\mathcal D,a}$, abbreviated to $K$.
Under $P$, the chain stays at its current state with
probability $1/2$. Otherwise, it chooses uniformly among
three rules, called $\mathsf{FORWARDS}$, $\mathsf{BACKWARDS}$,
and $\mathsf{RESET}$, and applies the selected rule.

To specify the rules, we need a pairing of tree states
and two distributions on lists of corridor labels.

The backward rule at a tree state uses the following pairing.
For $\sigma\in\Omega$, let
\[
\mathcal Q(\sigma)=\{b\in\mathcal D:\sigma\in B_b\}
\]
be the set of constraints violated by $\sigma$.
For $\sigma\in E_{\mathcal T}^{\star}$ with
$\mathcal Q(\sigma)\ne\{a\}$, let $d_\sigma(\mathcal T)$
be the last constraint in $\mathcal Q(\sigma)\setminus\{a\}$
examined by the selection procedure of \cref{sec:two-tree}
on input $\mathcal T$, and set
\[
\iota_\sigma(\mathcal T)
=\mathcal T\mathbin{\triangle}\{d_\sigma(\mathcal T)\}.
\]
By \cref{lem:terminal-involution}, the map
\[
\mathsf{T}_{\mathcal T}(\sigma)
\longmapsto \mathsf{T}_{\iota_\sigma(\mathcal T)}(\sigma)
\]
is a well-defined involution between tree states.
If $\mathcal Q(\sigma)=\{a\}$, the only tree state with
assignment $\sigma$ is $\mathsf{T}_{\{a\}}(\sigma)$, which is paired
with $\mathsf{I}(\sigma)$ instead.

The reset rule at an input state chooses a corridor label
from a random list $\mathbf L$.
As in \cref{sec:estimates}, independently draw
$\sigma_b\sim\nu_{\vbl(b)}$ for $b\in\mathcal D$.
For each $\mathcal T\in\mathfrak T_{\mathcal D,a}$, combine
$\sigma_t$, $t\in\mathcal T$, on their disjoint variable sets,
and complete the assignment to $U_{\mathcal T}$ with independent
product samples, using fresh randomness for each completion.
Write $\sigma_{\mathcal T}$ for the result, and set
\[
\mathbf L
=\{(\mathcal T,\sigma_{\mathcal T}):
  \mathcal T\in\mathfrak T_{\mathcal D,a},\
  \sigma_{\mathcal T}\in E_{\mathcal T}\}.
\]

The forward rule at $\mathsf{L}_{\ell(\mathcal T)}^e(\sigma)$
uses a list conditioned to contain $e=(\mathcal T,z)$.
To generate this list, denoted by $\mathbf L^e$, repeat the
experiment above with $\sigma_t=z|_{\vbl(t)}$ fixed for
$t\in\mathcal T$ and the completion for $\mathcal T$ fixed to
$z|_{U_{\mathcal T}\setminus\vbl(\mathcal T)}$.
All remaining random choices are independent and have their
original laws. By \cref{lem:return-list},
\[
\Law(\mathbf L^e)=\Law(\mathbf L\mid e\in\mathbf L).
\]
Choose an entry uniformly from $\mathbf L^e$. If the chosen
entry is $e$, move to $\mathsf{I}(\sigma)$; otherwise, remain at
$\mathsf{L}_{\ell(\mathcal T)}^e(\sigma)$.

The three rules are specified in \cref{alg:chain-step}.
A recursive call applies the rule already chosen for the
current step. For example, if $\mathsf{FORWARDS}$ was selected,
the recursive call applies the forward rule of the smaller
copy, using fresh randomness for any required samples.
It does not repeat the initial decision to stay put or
choose another rule. If the resulting state in the smaller
copy is $X'$, the next state of the full chain is $\psi(X')$,
using the embedding of that copy. The connections between the different types of states are
illustrated in \cref{fig:insertion-chain-schematic}.

Write $\mathsf{Step}(\mathcal D,a,s;\mathrm{rule})$ for applying a specified
rule $\mathrm{rule}$, and $\mathsf{Step}(\mathcal D,a,s)$ for a full step
of $P$, including the decision to stay put and the choice of rule.
In the algorithm, write $e=(\mathcal T,z)$ whenever a state carries a corridor label $e$.

\begin{algorithm}[htbp]
\small
\caption{Applying a rule:
  $\mathsf{Step}(\mathcal D,a,s;\mathrm{rule})$}
\label{alg:chain-step}
\KwIn{A state $s\in\mathsf S(\mathcal D,a)$ and a rule
$\mathrm{rule}\in\{\mathsf{FORWARDS},\mathsf{BACKWARDS},\mathsf{RESET}\}$.}
\KwOut{The resulting state.}

\uIf{$s=\mathsf{I}(\sigma)$}{
  \uIf{$\mathrm{rule}=\mathsf{FORWARDS}$}{
    \lIf{$\sigma\notin B_a$}{\Return $\mathsf{O}(\sigma)$}
    \Return $\mathsf{T}_{\{a\}}(\sigma)$\;
  }
  \ElseIf{$\mathrm{rule}=\mathsf{RESET}$}{
    Generate $\mathbf L$\;
    \If{$\mathbf L\ne\emptyset$}{
      Choose $e=(\mathcal T,z)$ uniformly from $\mathbf L$\;
      \Return $\mathsf{L}_{\ell(\mathcal T)}^e(\sigma)$\;
    }
  }
}
\uElseIf{$s=\mathsf{O}(\sigma)$}{
  \lIf{$\mathrm{rule}=\mathsf{BACKWARDS}$}{\Return $\mathsf{I}(\sigma)$}
}
\uElseIf{$s=\mathsf{T}_{\mathcal T}(\sigma)$}{
  \uIf{$\mathrm{rule}=\mathsf{FORWARDS}$}{
    Write $\sigma=(y,z)$ and draw
    $u\sim\nu_{U_{\mathcal T}}$\;
    Set $e=(\mathcal T,z)$\;
    \Return $\mathsf{L}_0^e(y,u)$\;
  }
  \ElseIf{$\mathrm{rule}=\mathsf{BACKWARDS}$}{
    \lIf{$\mathcal Q(\sigma)=\{a\}$}{\Return $\mathsf{I}(\sigma)$}
    \Return $\mathsf{T}_{\iota_\sigma(\mathcal T)}(\sigma)$\;
  }
}
\uElseIf{$s=\mathsf{L}_i^e(\sigma)$}{
  \uIf{$\mathrm{rule}=\mathsf{BACKWARDS}$}{
    \lIf{$i>0$}{\Return $\mathsf{L}_{i-1}^e(\sigma)$}
    Write $\sigma=(y,u)$\;
    \Return $\mathsf{T}_{\mathcal T}(y,z)$\;
  }
  \uElseIf{$i<\ell(\mathcal T)$}{
    $X'\gets\mathsf{Step}
      (\mathcal D_{\mathcal T,i+1},b_{\mathcal T,i+1},
       \mathsf{I}(\sigma);\mathrm{rule})$\;
    \Return $\psi_{e,i+1}(X')$\;
  }
  \ElseIf{$\mathrm{rule}=\mathsf{FORWARDS}$}{
    Generate $\mathbf L^e$\;
    Choose an entry $e'$ uniformly from this list\;
    \lIf{$e'=e$}{\Return $\mathsf{I}(\sigma)$}
  }
}
\ElseIf{$s=(e,i,X)$}{
  $X'\gets\mathsf{Step}
    (\mathcal D_{\mathcal T,i},b_{\mathcal T,i},X;\mathrm{rule})$\;
  \Return $\psi_{e,i}(X')$\;
}
\Return $s$\tcp*{In all other cases, the state is unchanged.}
\end{algorithm}

\begin{figure}[htbp]
\centering
\resizebox{\textwidth}{!}{
\begin{tikzpicture}[
 x=1cm,y=1cm,>={Latex[length=1.6mm]},
 every node/.style={font=\small},
 zone/.style={draw=black!25,rounded corners=5pt},
 state/.style={draw=blue!45!black,rounded corners=3pt,fill=blue!5,
   minimum width=2.4cm,minimum height=1cm,align=center},
 tree/.style={state,draw=orange!60!black,fill=orange!8},
 layer/.style={draw=green!40!black,fill=green!6,rounded corners=3pt,
   minimum width=2cm,minimum height=1cm,align=center},
 inside/.style={draw=black!30,fill=black!3,rounded corners=3pt,
   minimum width=2.2cm,minimum height=1cm,align=center},
 lab/.style={font=\scriptsize,align=center,fill=white,inner sep=3pt},
 move/.style={->,semithick}
]
\draw[zone] (0,7.9) rectangle (3.1,11.1);
\draw[zone] (4,6.5) rectangle (7.1,11.1);
\draw[zone] (8,7.9) rectangle (21.9,11.1);
\node[font=\small\bfseries] at (1.55,10.6) {Output};
\node[font=\small\bfseries] at (5.55,10.6) {Input};
\node[font=\small\bfseries] at (14.95,10.6) {Tree states};
\node[state] (R) at (1.55,9) {$\mathsf{O}(\sigma)$};
\node[state] (L) at (5.55,9) {$\mathsf{I}(\sigma)$};
\node[state] (Ltau) at (5.55,7.15) {$\mathsf{I}(\tau)$};
\node[tree] (S) at (12,9) {$\mathsf{T}_{\mathcal T}(\sigma)$};
\node[tree] (paired) at (19.5,9) {$\mathsf{T}_{\mathcal T'}(\sigma)$};
\draw[move] ([yshift=2mm]L.west) -- ([yshift=2mm]R.east);
\draw[move] ([yshift=-2mm]R.east) -- ([yshift=-2mm]L.west);
\node[lab] at (3.55,9.85) {$\sigma\notin B_a$};
\node[lab] at (3.55,8.15) {$\leftarrow$ F\quad B $\rightarrow$};
\draw[move] ([yshift=2mm]L.east) -- node[lab,above] {F: $\sigma\in B_a$, $\mathcal T=\{a\}$} ([yshift=2mm]S.west);
\draw[move] ([yshift=-2mm]S.west) -- node[lab,below] {B: $\mathcal Q(\sigma)=\{a\}$} ([yshift=-2mm]L.east);
\draw[<->,semithick] (S.east) -- node[lab,above] {B in both directions}
 node[lab,below] {$\mathcal T'=\iota_\sigma(\mathcal T)$\\$\mathcal Q(\sigma)\ne\{a\}$} (paired.west);

\node[draw=green!40!black,dashed,rounded corners=4pt,
 fill=green!3,align=center,text width=3.4cm,minimum height=1.1cm]
 (other) at (19.5,6.3)
 {Separate corridor\\[3pt]$(\mathcal T',z')$};
\draw[move] ([xshift=-2mm]paired.south) -- node[lab,left] {F} ([xshift=-2mm]other.north);
\draw[move] ([xshift=2mm]other.north) -- node[lab,right] {B} ([xshift=2mm]paired.south);
\node[font=\scriptsize,align=center,text width=3.5cm] at (19.5,5.25)
 {Its connection to the input is omitted.};

\draw[blue!55!black,line width=.8pt,rounded corners=5pt]
 (0,-.3) rectangle (8,2.6);
\draw[orange!65!black,line width=.8pt,dashed,rounded corners=5pt]
 (5.8,-.3) rectangle (13.8,2.6);
\draw[violet!65!black,line width=.8pt,rounded corners=5pt]
 (11.6,-.3) rectangle (19.6,2.6);
\node[text=blue!55!black] at (4,2.05) {Copy 3};
\node[text=orange!65!black] at (9.8,2.05) {Copy 2};
\node[text=violet!65!black] at (15.6,2.05) {Copy 1};
\node[layer] (C3) at (1.1,.6) {$\mathsf{L}_3^e(\tau)$};
\node[layer] at (6.9,.6) {$\mathsf{L}_2^e(\sigma_2)$};
\node[layer] at (12.7,.6) {$\mathsf{L}_1^e(\sigma_1)$};
\node[layer] (C0) at (18.5,.6) {$\mathsf{L}_0^e(y,u)$};
\node[inside] at (4,.6) {Internal states};
\node[inside] at (9.8,.6) {Internal states};
\node[inside] at (15.6,.6) {Internal states};

\draw[move] ([xshift=-3mm]Ltau.south) .. controls (3.5,6.5) and (.8,5.2) ..
 ([xshift=-3mm]C3.north);
\node[lab,text width=2.6cm] at (1.5,6.2)
 {RESET\\choose $e$\\uniformly from $\mathbf L$};
\draw[move] ([xshift=3mm]C3.north) .. controls (1.4,5) and (5.5,6.3) ..
 ([xshift=3mm]Ltau.south);
\node[lab,text width=3.4cm] at (5.3,4.65)
 {F: choose uniformly\\from $\mathbf L^e$; return\\if the entry is $e$};
\node[font=\scriptsize,align=center] at (6.2,3.35)
 {Both moves preserve the assignment $\tau$.};

\draw[move] ([xshift=-3mm]S.south) .. controls (11.7,5.7) and (16.2,3.6) ..
 ([xshift=-3mm]C0.north);
\node[lab,text width=2.8cm] at (11.6,5.55)
 {F: redraw $U_{\mathcal T}$\\$(y,z)\mapsto(y,u)$};
\draw[move] ([xshift=3mm]C0.north) .. controls (18.8,3.8) and (12.3,5.8) ..
 ([xshift=3mm]S.south);
\node[lab,text width=2.3cm] at (15.3,4.15)
 {B: restore $z$\\$(y,u)\mapsto(y,z)$};

\node[font=\small\bfseries] at (9.8,-1)
 {Corridor with fixed $e=(\mathcal T,z)$ and $\ell(\mathcal T)=3$};
\node[font=\scriptsize,align=center] at (9.8,-1.6)
 {F = FORWARDS,\quad B = BACKWARDS.\quad
 Each smaller copy uses its own transition rules.};
\end{tikzpicture}
}
\caption{The insertion chain, with one corridor expanded.
The two representative input states coincide when $\sigma=\tau$.
From $\mathsf{I}(\sigma)$, the forward
rule moves to $\mathsf{O}(\sigma)$ if $a$ is satisfied and to
$\mathsf{T}_{\{a\}}(\sigma)$ otherwise. Other tree states are paired as shown.
For the expanded corridor, $\sigma=(y,z)$ at the tree state;
the assignment may change along the corridor, reaching $\tau$
at its last layer. The list moves connect this last-layer state
to $\mathsf{I}(\tau)$ in the input copy. The paired tree has a separate
corridor with $z'=\sigma|_{U_{\mathcal T'}}$.
Outlined smaller copies share their boundary layers.
Holding moves and transitions inside smaller copies are omitted.}
\label{fig:insertion-chain-schematic}
\end{figure}
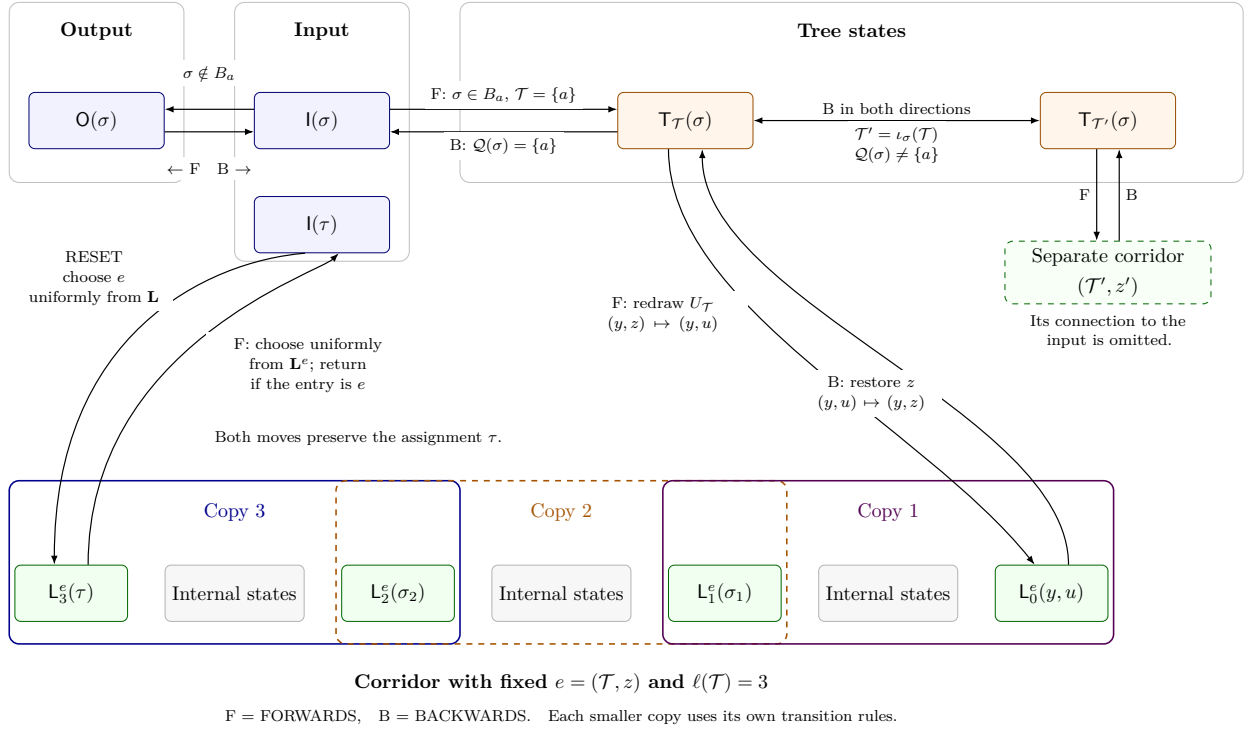

For $0<\beta<1$, define $K=K^\beta_{\mathcal D,a}$ by
changing only the transitions from the outer output states of
$P=P_{\mathcal D,a}$. Set
\[
K(\mathsf{O}(\sigma),\mathsf{I}(\sigma))=\frac{\beta}{6},
\qquad
K(\mathsf{O}(\sigma),\mathsf{O}(\sigma))=1-\frac{\beta}{6},
\]
and let $K$ agree with $P$ at every other state.
This modification increases the time spent at the output
states. In the next section, we choose $\beta$ so that
$K$ has a stationary distribution close to $\mu^+$ and
bound the number of steps required when the initial
distribution is $\mu^-$.

\section{Analysis of the insertion chain}\label{sec:insertion-proof}

We prove the following convergence bound for the kernel $K$
defined in \cref{sec:insertion}. Fix
$a\in\mathcal D\subseteq\mathcal C$, and write
$Z_-=Z_{\mathcal D^-}$.

\begin{proposition}\label{prop:measure}
For $0<\delta\leq1$, let
\begin{align*}
\beta&=\frac{\delta}{100(1+2m)}, \qquad
N=\left\lceil200{,}000(m+1)^2\delta^{-3}\right\rceil.
\end{align*}
Then
\[
\norm{\mu^-K^N-\mu^+}_{\TV}<\frac{\delta}{10}.
\]
\end{proposition}

\subsection{The pairing and the list distribution}

We first verify that the pairing move in \cref{alg:chain-step}
is well defined and that applying it twice returns to the
original state.

\begin{lemma}\label{lem:terminal-involution}
For $\mathcal T\in\mathfrak T_{\mathcal D,a}$ and
$\sigma\in E_{\mathcal T}^{\star}$ with $\mathcal Q(\sigma)\ne\{a\}$,
\[
\iota_\sigma(\mathcal T)\in\mathfrak T_{\mathcal D,a},
\qquad
\sigma\in E_{\iota_\sigma(\mathcal T)}^{\star},
\qquad
\iota_\sigma(\iota_\sigma(\mathcal T))=\mathcal T.
\]
Furthermore, $|\mathcal T|$ and $|\iota_\sigma(\mathcal T)|$ have opposite parity.
\end{lemma}

\begin{proof}
Since $\mathcal F(\mathcal T)=\mathcal F_1(\mathcal T)\sqcup\mathcal D_{\mathcal T}$,
\[
\sigma\in E_{\mathcal T}^{\star}
\quad\Longleftrightarrow\quad
\mathcal T\subseteq \mathcal Q(\sigma)\subseteq \mathcal T\cup \mathcal R(\mathcal T).
\]
Thus the selection procedure on $\mathcal T$ examines every constraint
in $\mathcal Q(\sigma)\setminus\{a\}$, so $d_\sigma(\mathcal T)$ is defined.

Write $d=d_\sigma(\mathcal T)$. Every accepted constraint belongs to
$\mathcal T\subseteq \mathcal Q(\sigma)$, and hence no constraint is accepted
after $d$ is examined. On input $\mathcal T\mathbin{\triangle}\{d\}$,
the procedure makes the same decisions before $d$, reverses
the decision at $d$, and rejects every subsequent candidate.
Its output is therefore $\mathcal T\mathbin{\triangle}\{d\}$, which is
a rooted $2$-tree by \cref{lem:canonical}.

By the time $d$ is examined, both runs have examined every
constraint in $\mathcal Q(\sigma)\setminus\{a\}$. The modified run accepts
only constraints in $\mathcal Q(\sigma)$, so
\[
\iota_\sigma(\mathcal T)\subseteq \mathcal Q(\sigma)
\subseteq\iota_\sigma(\mathcal T)\cup \mathcal R(\iota_\sigma(\mathcal T)).
\]
By the equivalence above,
$\sigma\in E_{\iota_\sigma(\mathcal T)}^{\star}$.
Moreover, $d$ is again the last examined constraint in
$\mathcal Q(\sigma)\setminus\{a\}$. Applying $\iota_\sigma$ a second time
therefore toggles $d$ again and returns $\mathcal T$.

Since $\mathcal T$ and $\iota_\sigma(\mathcal{T})$
differ at exactly $d$, this means they have opposite parity.
\end{proof}

The moves between the input and the corridors use a fresh
list in one direction and a conditioned list in the other.
To check detailed balance for these transitions, we need
the following relation between their distributions.

\begin{lemma}\label{lem:return-list}
For $\mathcal T\in\mathfrak T_{\mathcal D,a}$ and $z\in E_{\mathcal T}$, let
$e=(\mathcal T,z)$ and $q_e=\nu_{U_{\mathcal T}}(z)$. Then
\[
\mathbb P(e\in\mathbf L)=q_e,
\qquad
\Law(\mathbf L^e)=\Law(\mathbf L\mid e\in\mathbf L).
\]
\end{lemma}

\begin{proof}
Since $z\in E_{\mathcal T}$, the event $e\in\mathbf L$ is exactly
$\sigma_{\mathcal T}=z$. This fixes $\sigma_t=z|_{\vbl(t)}$ for each
$t\in \mathcal T$ and fixes the completion for $\mathcal T$ to
$z|_{U_{\mathcal T}\setminus\vbl(\mathcal T)}$. These assignments are independent,
so
\[
\mathbb P(e\in\mathbf L)
=\left(\prod_{t\in \mathcal T}\nu_{\vbl(t)}(z|_{\vbl(t)})\right)
\nu_{U_{\mathcal T}\setminus\vbl(\mathcal T)}(z|_{U_{\mathcal T}\setminus\vbl(\mathcal T)})
=\nu_{U_{\mathcal T}}(z).
\]
All other local assignments and completions remain independent
with their original laws. This is precisely the experiment
defining $\mathbf L^e$.
\end{proof}

Resampling or restoring $U_{\mathcal T}$ preserves satisfaction of
$\mathcal D_{\mathcal T}$, since $\vbl(\mathcal D_{\mathcal T})\cap U_{\mathcal T}=\emptyset$. By induction,
transitions within each smaller copy remain in its state
space. Thus every transition stays in
$\mathsf S(\mathcal D,a)$.

\subsection{The stationary distribution}
\label{subsec:retained-notation}

We show that $K$ is reversible with respect to a distribution
$\pi$ close to $\mu^+$.

Define positive weights $w$ on $\mathsf S(\mathcal D,a)$
by induction on $|\mathcal D|$. On the input, output,
and tree states, set
\[
w(\mathsf{I}(\sigma))=\nu(\sigma),
\qquad
w(\mathsf{O}(\sigma))=\nu(\sigma),
\qquad
w(\mathsf{T}_{\mathcal T}(\sigma))=\nu(\sigma).
\]
On the layer states and the internal states of smaller
copies, set
\[
\begin{aligned}
w(\mathsf{L}_i^e(\sigma))
&=\nu_{U_{\mathcal T}}(z)\nu(\sigma),\\
w((e,i,X))
&=\nu_{U_{\mathcal T}}(z)\,w(X).
\end{aligned}
\]
Here, $w(X)$ on the right denotes the recursively defined
weight of $X$ in
$\mathsf S(\mathcal D_{\mathcal T,i},b_{\mathcal T,i})$.
For every state $X$ in this smaller state space, including
its input and output states, we therefore have
\[
w(\psi_{e,i}(X))
=\nu_{U_{\mathcal T}}(z)\,w(X).
\]
In particular, the weights agree on layers shared by
adjacent copies.

Define
\[
\widetilde w(X)=
\begin{cases}
w(X)/\beta, & X\in\{\mathsf{O}(\sigma):\sigma\in\Good(\mathcal D)\},\\
w(X), & \text{otherwise},
\end{cases}
\]
and let
\[
M=\sum_{X\in\mathsf S(\mathcal D,a)}\widetilde w(X).
\]

\begin{lemma}\label{lem:stationary}
The distribution
\[
\pi(X)=\frac{\widetilde w(X)}M
\qquad(X\in\mathsf S(\mathcal D,a))
\]
is stationary for $K$, and $K$ is reversible with respect to $\pi$.
\end{lemma}

\begin{proof}
We first prove by induction on $|\mathcal D|$ that
\[
w(X)P(X,Y)=w(Y)P(Y,X).
\]
Call a pair of distinct states an edge if the transition
rules connect them. On a direct input--output edge
$\mathsf{I}(\sigma)\leftrightarrow \mathsf{O}(\sigma)$, a singleton edge
$\mathsf{I}(\sigma)\leftrightarrow \mathsf{T}_{\{a\}}(\sigma)$, or a pairing edge
$\mathsf{T}_{\mathcal T}(\sigma)\leftrightarrow
\mathsf{T}_{\mathcal T'}(\sigma)$, both endpoints have weight
$\nu(\sigma)$ and both transition probabilities are $1/6$.

On a resampling edge from $\mathsf{T}_{\mathcal T}(y,z)$ to
$\mathsf{L}_0^e(y,u)$, the forward probability is
$\nu_{U_{\mathcal T}}(u)/6$ and the reverse probability is $1/6$.
Since $\nu$ is a product distribution,
\[
\frac16\nu(y,z)\nu_{U_{\mathcal T}}(u)
=\frac16\nu_{U_{\mathcal T}}(z)\nu(y,u),
\]
which proves detailed balance on this edge.

For a return edge joining
$\mathsf{L}_{\ell(\mathcal T)}^e(\sigma)$ to $\mathsf{I}(\sigma)$,
let $e=(\mathcal T,z)$ and define
\[
\omega_e
=\mathbb E\left[
\frac{\mathbf1_{\{e\in\mathbf L\}}}{|\mathbf L|}
\right],
\]
where the fraction is zero when $\mathbf L$ is empty.
This is the probability of choosing $e$ from a fresh list.
Since $z\in E_{\mathcal T}$, the entry $e$ belongs to
$\mathbf L$ exactly when $\sigma_{\mathcal T}=z$. Thus
\[
q_e:=\mathbb P(e\in\mathbf L)=\nu_{U_{\mathcal T}}(z).
\]
By \cref{lem:return-list},
\[
\mathbb E\left[\frac1{|\mathbf L^e|}\right]
=\mathbb E\left[\frac1{|\mathbf L|}
  \,\middle|\,e\in\mathbf L\right]
=\frac{\omega_e}{q_e}.
\]
The probabilities of entering and leaving the corridor
along this edge are therefore $\omega_e/6$ and
$\omega_e/(6q_e)$, respectively. Hence
\[
\begin{aligned}
w(\mathsf{I}(\sigma))
P(\mathsf{I}(\sigma),\mathsf{L}_{\ell(\mathcal T)}^e(\sigma))
&=\nu(\sigma)\frac{\omega_e}{6}\\
&=q_e\nu(\sigma)\frac{\omega_e}{6q_e}\\
&=w(\mathsf{L}_{\ell(\mathcal T)}^e(\sigma))
P(\mathsf{L}_{\ell(\mathcal T)}^e(\sigma),\mathsf{I}(\sigma)).
\end{aligned}
\]

For two distinct states joined by an edge in a smaller copy,
the transition probabilities agree with those of the
auxiliary kernel $P$ for that smaller instance.
Both endpoint weights are multiplied by
$\nu_{U_{\mathcal T}}(z)$, so detailed balance follows
from the induction hypothesis.
When $|\mathcal D|=1$, there are no smaller copies, and
the checks above establish the base case.
Self-transitions satisfy detailed balance automatically.
Thus $w$ satisfies detailed balance for $P$.

We now pass to $K$. At each output state $\mathsf{O}(\sigma)$,
where $\sigma\in\Good(\mathcal D)$,
the probability of moving to $\mathsf{I}(\sigma)$ decreases from
$1/6$ to $\beta/6$, while the weight increases from
$\nu(\sigma)$ to $\nu(\sigma)/\beta$. Therefore
\[
\begin{aligned}
\widetilde w(\mathsf{O}(\sigma))K(\mathsf{O}(\sigma),\mathsf{I}(\sigma))
&=\frac{\nu(\sigma)}{\beta}\frac{\beta}{6}\\
&=\nu(\sigma)\frac16\\
&=\widetilde w(\mathsf{I}(\sigma))K(\mathsf{I}(\sigma),\mathsf{O}(\sigma)).
\end{aligned}
\]
On every other edge, the weights and transition probabilities
are unchanged. Thus $\widetilde w$ satisfies detailed balance
for $K$. Dividing by $M$ shows that $K$ is reversible with
respect to $\pi$, and hence $\pi$ is stationary.
\end{proof}

We next bound the distance between $\pi$ and $\mu^+$.
For $\sigma\in\Good(\mathcal D)$,
\[
\pi(\mathsf{O}(\sigma))=\frac{\nu(\sigma)}{\beta M},
\]
so $\pi$ conditioned on the output states is exactly $\mu^+$.
Consequently,
\[
\begin{aligned}
\norm{\pi-\mu^+}_{\TV}
&=1-\pi\bigl(\{\mathsf{O}(\sigma):\sigma\in\Good(\mathcal D)\}\bigr)\\
&=\frac1M\sum_{X\notin\{\mathsf{O}(\sigma):\sigma\in\Good(\mathcal D)\}}w(X).
\end{aligned}
\]
We therefore need a bound on the total weight outside
the output states.

For use in the mass and flow estimates, let
\[
\eta_{\mathcal T}
=\theta_{\mathcal T}\lambda_{\mathcal T}^2.
\]
The definition of $\lambda_{\mathcal T}$,
the inequalities
$Z_{\mathcal D}\leq Z_-\leq Z_{\mathcal D_{\mathcal T}}$,
and the deletion bound \eqref{eq:deletion-ratio} give
\[
\theta_{\mathcal T}
\frac{Z_{\mathcal D_{\mathcal T}}}{Z_-}
\leq \eta_{\mathcal T}
\leq
\theta_{\mathcal T}
\prod_{b\in N_{\mathcal D}^2[\mathcal T]}h_b^2
\leq \prod_{t\in\mathcal T}\rho_t.
\]
The last inequality uses
$\theta_{\mathcal T}\leq\prod_{t\in\mathcal T}p_t$,
$h_b^2\leq j_b$, and the definition of $\rho_t$.
Recall that
\[
\ell_x(\mathcal T)
=\sum_{i=1}^{\ell(\mathcal T)}x_{b_{\mathcal T,i}}.
\]
By \eqref{eq:tree-weight-sum},
\begin{equation}\label{eq:coefficient-sums}
\sum_{\mathcal T}\eta_{\mathcal T}\leq4cx_a,
\qquad
\sum_{\mathcal T}\eta_{\mathcal T}\ell_x(\mathcal T)<10cx_a.
\end{equation}

\begin{lemma}\label{lem:mass}
The total weight outside the output states satisfies
\[
\sum_{X\in\mathsf S(\mathcal D,a)\setminus\{\mathsf{O}(\sigma):\sigma\in\Good(\mathcal D)\}}w(X)
\leq (1+2mx_a)Z_-.
\]
\end{lemma}

\begin{proof}
We induct on $|\mathcal D|$, keeping $m=|\mathcal C|$ fixed.
For each corridor $e=(\mathcal T,z)$, count each smaller copy
without its output states, and count the last layer separately.
Each shared layer is then counted exactly once, as the input
of the next copy. Together with the input and tree states,
these pieces partition the states outside the output set
of the full construction.

The input states have total weight $Z_-$.
Fix $\mathcal T\in\mathfrak T_{\mathcal D,a}$.
Since $E_{\mathcal T}$ and the constraints in
$\mathcal D_{\mathcal T}$ involve disjoint variables,
the tree states $\mathsf{T}_{\mathcal T}(\sigma)$ have total weight
\[
\nu(E_{\mathcal T}^{\star})
=\theta_{\mathcal T}Z_{\mathcal D_{\mathcal T}}.
\]
For a fixed $z\in E_{\mathcal T}$, the last corridor layer
has weight $\nu_{U_{\mathcal T}}(z)Z_-$.
Summing over $z$, these layers have total weight
$\theta_{\mathcal T}Z_-$.

In copy $(e,i)$, the weights are those of the
smaller instance multiplied by $\nu_{U_{\mathcal T}}(z)$.
By induction, the weight of this copy excluding its output
states is at most
\[
\nu_{U_{\mathcal T}}(z)
Z_{\mathcal D_{\mathcal T,i-1}}
\bigl(1+2mx_{b_{\mathcal T,i}}\bigr).
\]
Summing over $z\in E_{\mathcal T}$ gives
\[
\theta_{\mathcal T}
Z_{\mathcal D_{\mathcal T,i-1}}
\bigl(1+2mx_{b_{\mathcal T,i}}\bigr).
\]

Both $Z_-$ and $Z_{\mathcal D_{\mathcal T,i-1}}$ are at most
$Z_{\mathcal D_{\mathcal T}}$. Thus the combined contribution
of the tree states and corridors associated with
$\mathcal T$ is at most
\[
\theta_{\mathcal T}Z_{\mathcal D_{\mathcal T}}
\bigl(2+\ell(\mathcal T)+2m\ell_x(\mathcal T)\bigr).
\]
Here the two constant terms account for the tree states
and the last corridor layers. Adding the input weight
and dividing by $Z_-$, we obtain
\[
\begin{aligned}
\frac1{Z_-}\sum_{X\notin\{\mathsf{O}(\sigma):\sigma\in\Good(\mathcal D)\}}w(X)
&\leq
1+\sum_{\mathcal T}
\eta_{\mathcal T}
\bigl(2+\ell(\mathcal T)+2m\ell_x(\mathcal T)\bigr)\\
&\leq
1+\bigl(4c(m+2)+20cm\bigr)x_a\\
&\leq 1+2mx_a.
\end{aligned}
\]
We used \eqref{eq:coefficient-sums},
$\ell(\mathcal T)\leq m$, $m\geq1$, and $c\leq1/100$.
When $|\mathcal D|=1$, there are no smaller copies,
so the same calculation establishes the base case.
\end{proof}

The partition in the proof also shows, by induction, that
the input sets of the full construction and all its
embedded copies are pairwise disjoint and lie outside
the output states of the full construction.
Indeed, within each smaller copy, its input states and
those of its nested copies avoid its output states.
They therefore lie in the piece assigned to that copy.
These pieces and the input of the full construction are
disjoint. By \cref{lem:mass}, the combined weight of all
these input states is at most $(1+2mx_a)Z_-$.

The output states have total weight $Z_{\mathcal D}/\beta$
under $\widetilde w$. Hence \cref{lem:mass} gives
\[
\begin{aligned}
M&=\frac{Z_{\mathcal D}}{\beta}
+\sum_{X\notin\{\mathsf{O}(\sigma):\sigma\in\Good(\mathcal D)\}}w(X)\\
&\leq
\frac{Z_{\mathcal D}}{\beta}+(1+2mx_a)Z_-.
\end{aligned}
\]
Using $(1-c)Z_-\leq Z_{\mathcal D}\leq Z_-$ and $x_a<1$,
we conclude that
\begin{equation}\label{eq:stationary-close}
\begin{aligned}
\frac M{Z_-}
&\leq \frac1\beta+1+2m,\\
\norm{\pi-\mu^+}_{\TV}
&=\frac1M\sum_{X\notin\{\mathsf{O}(\sigma):\sigma\in\Good(\mathcal D)\}}w(X)\\
&\leq
\frac{\beta(1+2mx_a)Z_-}{Z_{\mathcal D}}
\leq \frac{\beta(1+2m)}{1-c}.
\end{aligned}
\end{equation}

\subsection{The flow}\label{subsec:flow-proof}

We next construct a signed flow from $\mu^-$ to $\mu^+$.
Its energy will control
$\norm{(\mu^- -\mu^+)K^N}_{\TV}$ in the proof of
\cref{prop:measure}.

A \emph{signed flow} assigns a real number $\varphi(X,Y)$
to each oriented edge, with
$\varphi(Y,X)=-\varphi(X,Y)$.
Its divergence at $X$ is
\[
\partial\varphi(X)
=\sum_{Y:\{X,Y\}\text{ is an edge}}\varphi(X,Y).
\]
For an edge $\{X,Y\}$, define its capacity by
\[
\operatorname{cap}(X,Y)
=w(X)P(X,Y)
=\widetilde w(X)K(X,Y).
\]
By detailed balance, this does not depend on the order
of $X$ and $Y$. The capacities are independent of $\beta$.
In sums indexed by $\{X,Y\}$, each edge is counted once,
with either ordering of its endpoints.
The energy of $\varphi$ is
\[
\sum_{\{X,Y\}}
\frac{\varphi(X,Y)^2}{\operatorname{cap}(X,Y)}.
\]

\begin{lemma}\label{lem:flow}
There is a signed flow $\varphi$ on the edges of $P$ satisfying
\[
\partial\varphi=\mu^- -\mu^+
\]
and
\[
\sum_{\{X,Y\}}
\frac{\varphi(X,Y)^2}{\operatorname{cap}(X,Y)}
\leq\frac{6/(1-c)+mx_a}{Z_-}.
\]
\end{lemma}

\begin{proof}
We induct on $|\mathcal D|$, keeping $m=|\mathcal C|$ fixed.
For a smaller instance $(\mathcal E,b)$, let
$\varphi_{\mathcal E,b}$ be the flow given by induction.
For $e=(\mathcal T,z)$ with
$\mathcal T\in\mathfrak T_{\mathcal D,a}$ and
$z\in E_{\mathcal T}$, recall that
$q_e=\nu_{U_{\mathcal T}}(z)$, and let
\[
\gamma_e=\lambda_{\mathcal T}q_e.
\]
We use this coefficient to scale the flow through the
corridor $e$ and its smaller copies.

On each edge of the indicated type, define
\[
\begin{aligned}
\varphi(\mathsf{I}(\sigma),\mathsf{O}(\sigma))
&=\frac{\nu(\sigma)}{Z_{\mathcal D}},\\
\varphi(\mathsf{I}(\sigma),\mathsf{T}_{\{a\}}(\sigma))
&=\frac{\nu(\sigma)}{Z_{\mathcal D}},\\
\varphi(\mathsf{T}_{\mathcal T}(\sigma),
        \mathsf{T}_{\iota_\sigma(\mathcal T)}(\sigma))
&=(-1)^{|\mathcal T|}
  \frac{\nu(\sigma)}{Z_{\mathcal D}},\\
\varphi(\mathsf{T}_{\mathcal T}(y,z),
        \mathsf{L}_0^e(y,u))
&=\gamma_e\mu_{\mathcal D_{\mathcal T}}(y,u),\\
\varphi(\mathsf{L}_{\ell(\mathcal T)}^e(\sigma),
        \mathsf{I}(\sigma))
&=\gamma_e\mu_{\mathcal D^-}(\sigma).
\end{aligned}
\]
In each smaller copy $(e,i)$, set
\[
\varphi(\psi_{e,i}(X),
        \psi_{e,i}(Y))
=\gamma_e
 \varphi_{\mathcal D_{\mathcal T,i},b_{\mathcal T,i}}(X,Y).
\]
Reverse orientations receive the negative values.
These definitions are consistent because paired trees
have opposite parity and distinct copies joining the layers
share no edges. When $|\mathcal D|=1$, there are no smaller
copies.

We first check the divergence.
Since $\mathcal D_{\mathcal T}$ involves no variables in
$U_{\mathcal T}$, the definition of $\lambda_{\mathcal T}$
gives
\[
\sum_{u\in\Omega_{U_{\mathcal T}}}
\varphi(\mathsf{T}_{\mathcal T}(y,z),
        \mathsf{L}_0^e(y,u))
=(-1)^{|\mathcal T|-1}
  \frac{\nu(y,z)}{Z_{\mathcal D}}.
\]
At the tree state $\mathsf{T}_{\mathcal T}(y,z)$, this cancels
the contribution from the pairing or singleton edge.

At a shared layer $\mathsf{L}_i^e$,
where $1\leq i<\ell(\mathcal T)$, copies $i$ and $i+1$
contribute
\[
-\gamma_e
 \mu_{\mathcal D_{\mathcal T,i}}(\sigma)
\quad\text{and}\quad
\gamma_e
 \mu_{\mathcal D_{\mathcal T,i}}(\sigma),
\]
respectively, so their contributions cancel.
At the first and last layers, the contributions from
the smaller copies cancel those from the resampling
and return edges. If $\ell(\mathcal T)=0$, the resampling
and return contributions cancel directly.
Internal states of smaller copies have zero divergence
by induction.

Finally, write $r=r_{\mathcal D,a}$.
Using
$\sum_{\mathcal T}\lambda_{\mathcal T}\theta_{\mathcal T}
=r/(1-r)$ and $Z_{\mathcal D}=(1-r)Z_-$, we obtain
\[
\begin{aligned}
\partial\varphi(\mathsf{I}(\sigma))
&=\frac{\nu(\sigma)}{Z_{\mathcal D}}
-\mu_{\mathcal D^-}(\sigma)
 \sum_{\mathcal T}\lambda_{\mathcal T}\theta_{\mathcal T}\\
&=\mu_{\mathcal D^-}(\sigma).
\end{aligned}
\]
At each output state,
$\partial\varphi(\mathsf{O}(\sigma))=-\mu_{\mathcal D}(\sigma)$.
This proves the divergence identity.

We now bound the energy.
Recall that
$\eta_{\mathcal T}=\theta_{\mathcal T}\lambda_{\mathcal T}^2$.
A singleton or pairing edge incident to
$\mathsf{T}_{\mathcal T}(\sigma)$ has energy
$6\nu(\sigma)/Z_{\mathcal D}^2$.
Summing over tree states counts every singleton edge once
and every pairing edge twice, so their total energy is
at most
\[
\frac6{Z_{\mathcal D}^2}
\sum_{\mathcal T}\theta_{\mathcal T}
                    Z_{\mathcal D_{\mathcal T}}
=6\sum_{\mathcal T}
  \frac{\eta_{\mathcal T}}{Z_{\mathcal D_{\mathcal T}}}.
\]
For $e=(\mathcal T,z)$, a resampling edge from
$\mathsf{T}_{\mathcal T}(y,z)$ to $\mathsf{L}_0^e(y,u)$ has
capacity $q_e\nu(y,u)/6$.
Summing its energy over $y,u$ gives
$6\lambda_{\mathcal T}^2q_e/Z_{\mathcal D_{\mathcal T}}$.
Thus all resampling edges together have energy
$6\sum_{\mathcal T}\eta_{\mathcal T}/Z_{\mathcal D_{\mathcal T}}$.
The input--output edges contribute $6/Z_{\mathcal D}$.
Combining these bounds gives
\[
\frac6{Z_{\mathcal D}}
+12\sum_{\mathcal T}
   \frac{\eta_{\mathcal T}}{Z_{\mathcal D_{\mathcal T}}}
\leq
\frac1{Z_-}
\left(\frac6{1-c}+12\sum_{\mathcal T}\eta_{\mathcal T}\right).
\]

For a fixed entry $e=(\mathcal T,z)$, the return edge at
$\sigma$ has capacity $\omega_e\nu(\sigma)/6$.
Summing over $\sigma$, the return edges for this entry
have total energy
\[
\frac{6\lambda_{\mathcal T}^2q_e^2}{\omega_e Z_-}.
\]
Since $q_e=\mathbb P(e\in\mathbf L)$,
Cauchy--Schwarz gives
\[
\begin{aligned}
q_e^2
&\leq
\mathbb E\left[
\frac{\mathbf1_{\{e\in\mathbf L\}}}{|\mathbf L|}
\right]
\mathbb E\left[
\mathbf1_{\{e\in\mathbf L\}}|\mathbf L|
\right]\\
&=\omega_e\,
\mathbb E\left[
\mathbf1_{\{e\in\mathbf L\}}|\mathbf L|
\right].
\end{aligned}
\]
Multiplying by $\lambda_{\mathcal T}^2/\omega_e$
and summing over entries, we obtain
\[
\begin{aligned}
\sum_{e=(\mathcal T,z)}
\frac{\lambda_{\mathcal T}^2q_e^2}{\omega_e}
&\leq
\mathbb E\left[
|\mathbf L|
\sum_{(\mathcal T,z)\in\mathbf L}\lambda_{\mathcal T}^2
\right]\\
&=
\sum_{\mathcal T,\mathcal T'}
\lambda_{\mathcal T}^2
\mathbb P(\sigma_{\mathcal T}\in E_{\mathcal T},\
          \sigma_{\mathcal T'}\in E_{\mathcal T'}).
\end{aligned}
\]
The joint event requires every constraint in
$\mathcal T\cup\mathcal T'$ to be active.
Independence of the local assignments $\sigma_t$,
together with the deletion bound, therefore gives
\[
\begin{aligned}
&\lambda_{\mathcal T}^2
\mathbb P(\sigma_{\mathcal T}\in E_{\mathcal T},\
          \sigma_{\mathcal T'}\in E_{\mathcal T'})\\
&\qquad\leq
\left(\prod_{b\in N_{\mathcal D}^2[\mathcal T]}j_b\right)
\left(\prod_{t\in\mathcal T\cup\mathcal T'}p_t\right)
\leq
\prod_{t\in\mathcal T\cup\mathcal T'}\rho_t.
\end{aligned}
\]
Applying the pair bound following \cref{lem:connected-sets} and \eqref{eq:pair-sum}, we conclude that
\[
\sum_{e=(\mathcal T,z)}
\frac{\lambda_{\mathcal T}^2q_e^2}{\omega_e}
\leq
\sum_{\mathcal T,\mathcal T'}
\prod_{t\in\mathcal T\cup\mathcal T'}\rho_t
\leq\frac{8c}{3}x_a.
\]
The return edges therefore contribute at most
$16cx_a/Z_-$.

It remains to bound the energy in the smaller copies.
In copy $(e,i)$, capacities are multiplied
by $q_e$ and flows by $\gamma_e$.
Thus energy is multiplied by
\[
\frac{\gamma_e^2}{q_e}
=\lambda_{\mathcal T}^2q_e.
\]
For fixed $\mathcal T$, summing over $z$ and $i$
and applying the induction hypothesis bounds this
contribution by
\[
\begin{aligned}
\eta_{\mathcal T}
\sum_{i=1}^{\ell(\mathcal T)}
\frac{6/(1-c)+mx_{b_{\mathcal T,i}}}
     {Z_{\mathcal D_{\mathcal T,i-1}}}
&\leq
\frac{\eta_{\mathcal T}}{Z_-}
\left(\frac{6\ell(\mathcal T)}{1-c}
      +m\ell_x(\mathcal T)\right).
\end{aligned}
\]
Here we used
$Z_{\mathcal D_{\mathcal T,i-1}}\geq Z_-$ and
$\ell_x(\mathcal T)
=\sum_{i=1}^{\ell(\mathcal T)}x_{b_{\mathcal T,i}}$.

Adding all contributions and using
\eqref{eq:coefficient-sums} and $\ell(\mathcal T)\leq m$,
we obtain
\[
\begin{aligned}
Z_-\sum_{\{X,Y\}}
\frac{\varphi(X,Y)^2}{\operatorname{cap}(X,Y)}
&\leq
\frac6{1-c}
+12\sum_{\mathcal T}\eta_{\mathcal T}
+16cx_a\\
&\qquad+
\sum_{\mathcal T}\eta_{\mathcal T}
\left(\frac{6\ell(\mathcal T)}{1-c}
      +m\ell_x(\mathcal T)\right)\\
&\leq
\frac6{1-c}
+\left[64c+
 \left(\frac{24c}{1-c}+10c\right)m\right]x_a\\
&\leq
\frac6{1-c}+mx_a.
\end{aligned}
\]
The last inequality uses $m\geq1$ and
$c(74+24/(1-c))<1$ for $c\leq1/100$.
This proves the energy bound and completes the induction.
\end{proof}

\subsection{Convergence}
\label{subsec:measure-estimate}

We now complete the proof of \cref{prop:measure} by bounding
$\norm{\mu^-K^N-\mu^+}_{\TV}$.

\begin{proof}[Proof of \cref{prop:measure}]
Since $\pi K^N=\pi$, contraction of total variation and
the triangle inequality give
\[
\begin{aligned}
\norm{\mu^-K^N-\mu^+}_{\TV}
&\leq
\norm{(\mu^- -\mu^+)K^N}_{\TV}
+\norm{\mu^+K^N-\mu^+}_{\TV}\\
&\leq
\norm{(\mu^- -\mu^+)K^N}_{\TV}
+2\norm{\pi-\mu^+}_{\TV}.
\end{aligned}
\]
By \eqref{eq:stationary-close} and
$\beta=\delta/[100(1+2m)]$,
\[
\norm{\pi-\mu^+}_{\TV}
\leq\frac{\delta}{100(1-c)}
\leq\frac{\delta}{99}.
\]
To bound the remaining term, we apply the flow estimate
to test functions.

For a signed measure $\xi$ and a real function $f$ on
$\mathsf S(\mathcal D,a)$, write
$\xi(f)=\sum_X\xi(X)f(X)$.
Define the Dirichlet form
\[
\begin{aligned}
\mathcal{E}_K(f,f)
&=\frac12\sum_{X,Y}
\pi(X)K(X,Y)(f(X)-f(Y))^2\\
&=\frac1M\sum_{\{X,Y\}}
\operatorname{cap}(X,Y)(f(X)-f(Y))^2.
\end{aligned}
\]
Let $\varphi$ be the flow from \cref{lem:flow}.
Since $\partial\varphi=\mu^- -\mu^+$, summing over one
orientation of each edge and applying Cauchy--Schwarz gives
\[
\begin{aligned}
|(\mu^- -\mu^+)(f)|^2
&=\left|
\sum_{\{X,Y\}}
\varphi(X,Y)(f(X)-f(Y))
\right|^2\\
&\leq
\left(\sum_{\{X,Y\}}
\frac{\varphi(X,Y)^2}{\operatorname{cap}(X,Y)}\right)
\left(\sum_{\{X,Y\}}
\operatorname{cap}(X,Y)(f(X)-f(Y))^2\right)\\
&\leq
\frac M{Z_-}
\left(\frac6{1-c}+mx_a\right)\mathcal{E}_K(f,f).
\end{aligned}
\]

Apply this inequality to $f=K^Ng$, where
$\norm g_\infty\leq1$ and
\[
(K^Ng)(X)=\sum_Y K^N(X,Y)g(Y).
\]
The kernel $K$ is lazy and reversible, so its eigenvalues
lie in $[0,1]$. Using the spectral decomposition and
$(1-s)s^{2N}\leq1/(2N+1)$ for $0\leq s\leq1$, we obtain
\[
\mathcal{E}_K(K^Ng,K^Ng)
\leq\frac{\norm g_{L^2(\pi)}^2}{2N+1}
\leq\frac1{2N+1}.
\]
Since
\[
((\mu^- -\mu^+)K^N)(g)
=(\mu^- -\mu^+)(K^Ng),
\]
the preceding bounds give
\[
\begin{aligned}
\norm{(\mu^- -\mu^+)K^N}_{\TV}
&=\frac12\sup_{\norm g_\infty\leq1}
|(\mu^- -\mu^+)(K^Ng)|\\
&\leq\frac12
\sqrt{\frac{M}{Z_-(2N+1)}
\left(\frac6{1-c}+mx_a\right)}.
\end{aligned}
\]

By \eqref{eq:stationary-close} and the choice of $\beta$,
\[
\begin{aligned}
\frac M{Z_-}\left(\frac6{1-c}+mx_a\right)
&\leq
\left(\frac1\beta+1+2m\right)
\left(\frac6{1-c}+m\right)\\
&\leq\frac{1500(m+1)^2}{\delta}.
\end{aligned}
\]
Since $N\geq200000(m+1)^2\delta^{-3}$, it follows that
\[
\norm{(\mu^- -\mu^+)K^N}_{\TV}<\frac{\delta}{20}.
\]
Combining the two total-variation bounds, we conclude that
\[
\norm{\mu^-K^N-\mu^+}_{\TV}
<\frac{\delta}{20}+\frac{2\delta}{99}
<\frac{\delta}{10}.
\]
\end{proof}

\section{The sampling algorithm}\label{sec:implementation}

We sample from $\mu_{\mathcal C}$ by inserting the constraints one
at a time. Write $\mathcal C_i=\{a_1,\ldots,a_i\}$, with
$\mathcal C_0=\emptyset$, for the fixed ordering of the constraints.
The algorithm first finds a satisfying assignment $\sigma_\star$,
which every insertion uses as its fallback output.

\begin{algorithm}[H]
\caption{$\Sample(\mathcal C,\eps)$}\label{alg:sample}
\KwIn{A constraint family $\mathcal C=\{a_1,\ldots,a_m\}$ and
accuracy $0<\eps\leq1$.}
\KwOut{$\sigma\in\Good(\mathcal C)$ with
$\norm{\Law(\sigma)-\mu_{\mathcal C}}_{\TV}\leq\eps$.}
Find $\sigma_\star\in\Good(\mathcal C)$ by the Moser--Tardos
algorithm\;
Independently draw $\sigma\sim\nu$\;
$\delta\gets\eps/(m+1)$\;
\For{$i=1,\ldots,m$}{
  $\sigma\gets\Insert(\mathcal C_i,a_i,\sigma,\delta)$,
  using fallback $\sigma_\star$\;
}
\Return{$\sigma$}\;
\end{algorithm}

The insertion guarantee yields the desired sampling law by the
telescoping argument following \cref{prop:insert}. Conditional on
$\sigma_\star$, the initial assignment still has law
$\nu=\mu_{\mathcal C_0}$, since it is generated independently.
The final error is therefore at most $m\delta\leq\eps$ for every
$\sigma_\star$, and averaging over $\sigma_\star$ preserves this
bound. We now implement the insertion routine and bound its
running time, completing the proof of \cref{prop:insert}.

\subsection{The insertion routine}

Fix $a\in\mathcal D\subseteq\mathcal C$ and $0<\delta\leq1$.
Starting from $\mathsf{I}(\sigma)$, we run the kernel $K$ for $N$ steps
and return the assignment at the final state if it is an output
state. Otherwise, we return $\sigma_\star$. We also return
$\sigma_\star$ if the simulation exceeds its time limit.

The parameters are specified in \cref{alg:insert}, where
$m=|\mathcal C|$. The polynomial $T_{\mathrm{step}}(I)$ bounds the expected
cost of each step when the initial assignment has law
$\mu_{\mathcal D\setminus\{a\}}$, as proved below.

\begin{algorithm}[H]
\caption{$\Insert(\mathcal D,a,\sigma,\delta)$}\label{alg:insert}
\KwIn{$a\in\mathcal D$, an assignment
$\sigma\in\Good(\mathcal D\setminus\{a\})$, accuracy
$0<\delta\leq1$, and the shared fallback
$\sigma_\star\in\Good(\mathcal C)$.}
\KwOut{An assignment satisfying every constraint in $\mathcal D$.}
$\beta\gets\delta/[100(1+2m)]$\;
$N\gets\left\lceil200000(m+1)^2\delta^{-3}\right\rceil$\;
$t_{\max}\gets\left\lceil20NT_{\mathrm{step}}(I)/\delta\right\rceil$\;
$X\gets \mathsf{I}(\sigma)$; start a timer\;
\For{$j=1,\ldots,N$}{
  Draw $X'\sim K(X,\mathord\cdot)$ and set $X\gets X'$,
  interrupting the simulation if the timer reaches $t_{\max}$\;
  \If{the time limit is reached}{
    \Return{$\sigma_\star$}\;
  }
}
\If{$X=\mathsf{O}(\tau)$ is an output state}{
  \Return{$\tau$}\;
}
\Return{$\sigma_\star$}\;
\end{algorithm}

\subsection{Running time and correctness}

We first bound the expected cost of the unrestricted simulation
when the input assignment has law
$\mu_{\mathcal D\setminus\{a\}}$. Write $\mu^-$ for this law
on the input states and $\mu^+$ for $\mu_{\mathcal D}$ on the
output states, and set $Z_- = Z_{\mathcal D\setminus\{a\}}$.
A state consists of its assignment, its type, and the labels
$(e,i)$ of the copies containing it. There are at most
$m$ such labels, since the constraint family shrinks at each level.
The state updates, pairing, constraint tests, and product
resamplings therefore take polynomial time in $I$. The remaining
cost is generating the fresh and conditioned return lists.

The reversible weights provide a bound that holds throughout the
simulation. Since $\mu^-\leq\widetilde w/Z_-$ pointwise and
$\widetilde wK=\widetilde w$,
\[
\mu^-K^t\leq\frac{\widetilde w}{Z_-}
\qquad(t\geq0).
\]
List generation occurs outside the outer output states,
where $\widetilde w=w$.
We can therefore bound its cost by summing over these states
with weights $w(X)/Z_-$.

Consider a copy with parameters $(\mathcal E,b)$. In its list
experiment, each constraint is active when its independent
local sample violates it. Let $A$ be the number of active
connected sets containing $b$ in $G_{\mathcal E}^2$.
The enumeration procedure visits these sets with polynomial
work per set. Generating a fresh list $\mathbf L$ and choosing
a uniform entry therefore takes expected time at most
$I^{O(1)}\mathbb E[1+A]$.

For an entry $e=(\mathcal T,z)$, write
$q_e=\nu_{U_{\mathcal T}}(z)$. By \cref{lem:return-list},
$q_e=\mathbb P(e\in\mathbf L)$, and the pinned experiment
generates the list conditioned on $e\in\mathbf L$.
The same enumeration bound gives expected time at most
$I^{O(1)}\mathbb E[1+A\mid e\in\mathbf L]$.
This conditional expectation need not be uniformly small.
However, the weights of the states where this operation is
performed contain a factor $q_e$, so we need to bound only
the sum weighted by these probabilities.

The estimates
$\mathbb EA\leq\mathbb EA^2\leq2cx_b$
from \cref{lem:tree-bounds}, together with $|\mathbf L|\leq A$,
give
\[
\begin{aligned}
\mathbb E[1+A]&\leq1+2cx_b,\\
\sum_e q_e\mathbb E[1+A\mid e\in\mathbf L]
&=\mathbb E[|\mathbf L|(1+A)]\\
&\leq\mathbb E[A+A^2]\leq4cx_b.
\end{aligned}
\]
If $\kappa$ multiplies this copy's weights in the full chain,
its input states have total weight
$\kappa Z_{\mathcal E\setminus\{b\}}$. The last layer of
corridor $e$, where the conditioned list is generated, has
total weight $\kappa q_eZ_{\mathcal E\setminus\{b\}}$.
Thus the weighted expected cost of both types of list operation
in this copy is at most
$\kappa Z_{\mathcal E\setminus\{b\}}I^{O(1)}$.

A transition may evaluate the chosen rule in several nested
copies. We bound its list-generation cost by summing the
costs over all copies in which a list operation could occur.
By \cref{lem:mass}, the input sets of all copies, including
the full construction itself, are disjoint and have combined
weight at most $(1+2mx_a)Z_-$.
Summing the preceding bound over copies and dividing by $Z_-$
therefore gives a polynomial bound on the expected list cost
of each step.

Including the other operations, choose a fixed polynomial
$T_{\mathrm{step}}(I)$ that bounds the expected cost of a step uniformly in
$t$, $\mathcal D$, $a$, and $\beta$. This polynomial depends
only on the input representation and the given bounds on
oracle costs. The unrestricted $N$-step run has expected cost
at most $NT_{\mathrm{step}}(I)$. The time limit in \cref{alg:insert} is
polynomial in $I$ and $1/\delta$, and Markov's inequality gives
\[
\mathbb P_{\mu^-}(\text{the simulation is interrupted})
\leq\frac{NT_{\mathrm{step}}(I)}{t_{\max}}\leq\frac\delta{20}.
\]

\begin{proof}[Proof of \cref{prop:insert}]
Fix $\sigma_\star$. Let $F$ map $\mathsf{O}(\sigma)$ to $\sigma$ and
every other state to $\sigma_\star$, so that
$\mu^+F=\mu_{\mathcal D}$. For the values of $\beta$ and $N$
used in \cref{alg:insert}, \cref{prop:measure} gives
$\norm{\mu^-K^N-\mu^+}_{\TV}<\delta/10$.
Couple the algorithm with the unrestricted simulation until
the time limit is reached. For the returned assignment $\sigma$,
contraction of total variation gives
\[
\begin{aligned}
\norm{\Law(\sigma)-\mu_{\mathcal D}}_{\TV}
&\leq\norm{\mu^-K^NF-\mu^+F}_{\TV}
  +\mathbb P_{\mu^-}(\text{the simulation is interrupted})\\
&<\frac\delta{10}+\frac\delta{20}<\delta.
\end{aligned}
\]
Every output satisfies $\mathcal D$, and the time limit bounds
the simulation time for every input assignment.
\end{proof}

It remains to account for finding $\sigma_\star$.
The Moser--Tardos algorithm \cite{moserTardos2010}
applies with parameters $cx_b$, since
\[
p_b
\leq cx_b\prod_{d\in N_G^2(b)}(1-x_d)
\leq cx_b\prod_{d\in N_G(b)}(1-cx_d).
\]
Its expected number of resamplings is at most
\[
\sum_{b\in\mathcal C}\frac{cx_b}{1-cx_b}
\leq\frac{cm}{1-c},
\]
so this preliminary run takes polynomial expected time.
\Cref{alg:sample} performs it once and then makes $m$
insertions with $\delta=\eps/(m+1)$. Its expected running
time is therefore polynomial in $I$ and $1/\eps$.
Together with the error bound at the start of this section,
this proves the sampling assertion of \cref{thm:main}.

\bibliographystyle{amsplain0}
\bibliography{main}
\end{document}